\documentclass[12pt, leqno]{article}

\usepackage{chngcntr}
\usepackage{oxford3}
\usepackage[applemac]{inputenc}
\usepackage[T1]{fontenc}
\usepackage{amssymb}
\usepackage{amsmath}
\usepackage{color}
\usepackage{enumerate}
\usepackage{geometry}
\usepackage{float}
\usepackage{listings}
\usepackage{subfigure}
\usepackage{bm}
\usepackage{eurosym}
\usepackage{appendix}
\usepackage[english]{babel}
\usepackage{amsthm}
\usepackage{hyperref}
\usepackage{cleveref}
\usepackage{tikz}
\usepackage{mathtools}
\usepackage{multirow}

\DeclarePairedDelimiter\abs{\lvert}{\rvert}%
\DeclarePairedDelimiter\norm{\lVert}{\rVert}%

\numberwithin{equation}{section}

\newtheorem{Theorem}{Theorem}[section]
\newtheorem{Lemma}{Lemma}[section]
\newtheorem{Corollary}{Corollary}[section]

\newtheorem{Proposition}{Proposition}[section]

\newtheorem{Assumption}{Assumption}[section]

\newcounter{defcounter}
\newenvironment{Mequation}{%
\addtocounter{equation}{-1}
\refstepcounter{defcounter}

\begin{equation}}
{\end{equation}}

\newcounter{defcounterr}
\newcounter{defcounterrr}
\newcounter{defcounterrrr}
\newcounter{defcounterrrrr}
\newcommand{\bW}{\mathbf{W}}

\newcommand{\bSig}{\mathbf{\Sigma}}

 \newcommand \ag{\alpha}

\renewcommand{\baselinestretch}{1.05}
\begin{document}

\title{Detection of periodicity  in functional time series}
\author{Siegfried H\"ormann$^1$\thanks{Corresponding author.
Email: shormann@ulb.ac.be}
\and Piotr Kokoszka$^2$
\and Gilles Nisol$^{3}$}
\date{\today}
\maketitle
\vspace{-1cm}
$^{1}$ Department of Mathematics, Universit\'e libre de Bruxelles
CP210, Bd.\ du Triomphe, B-1050 Brussels, Belgium.

$^2$ Department of  Statistics,
Colorado State University, Fort Collins, CO 80523-1877, USA.

$^{3}$ ECARES, Universit\'e libre de Bruxelles, 
50 Avenue Franklin Roosevelt, B-1050 Brussels, Belgium.
\begin{abstract}
We derive several tests for the presence of a periodic component
in a time series of functions.  We consider both the traditional
setting in which the periodic functional signal is contaminated by
functional white noise, and a more general setting of a contaminating
process which is weakly dependent. Several forms of the periodic
component are considered. 
Our tests are motivated by the likelihood principle and  fall into two broad categories, which
we term multivariate and fully functional. Overall, for the functional
series that motivate this research, the fully functional tests exhibit a
superior balance of size and power. Asymptotic null distributions of
all tests are derived and their consistency is established. Their
finite sample performance is examined and compared by
numerical studies and application to pollution data.
\end{abstract}

\\[1ex]

\footnotetext{
This research was supported by the
Communaut\'e fran\c{c}aise de Belgique---Actions de Recherche
Concert\'ees (2010--2015), Interuniversity Attraction Poles Programme (IAP-network P7/06) of the Belgian Science Policy Office, 
by  United States NSF grant DMS--1462067 and by the Fonds de la Recherche Scientifique - FNRS under Grant MCF/FC 24535233.}

\section{Introduction}\label{s:i}
Periodicity is one of the most important characteristics of time
series, and tests for periodicity go back to the very origins of the
field, e.g.\ \citetext{schuster:1898}, \citetext{walker:1914},
\citetext{fisher:1929}, \citetext{jenkins:priestley:1957},
\citetext{hannan:1961}, among many others.
An excellent account of these early developments is given in Chapter
10 of \citetext{brockwell:davis:1991}.

We respond to the need to develop periodicity tests for time series of
functions---short
\emph{functional time series (FTS's)}.  Examples of
FTS's include annual temperature or smoothed precipitation curves,
e.g. \citetext{gromenko:kokoszka:reimherr:2016}, daily pollution level
curves, e.g. \citetext{aue:dubartNorinho:hormann:2015}, various daily
curves derived from high frequency asset price data,
e.g. \citetext{horvath:kokoszka:rice:2014}, yield curves,
e.g. \citetext{hays:shen:huang:2012}, daily vehicle traffic curves,
e.g. \citetext{klepsch:kluppelberg:wei:2016}.  More complex objects,
like sequences of 2D satellite images or 3D brain scans, have also
been considered, e.g. \citetext{jun:stein:2008} and
\citetext{sarty:2007}, but the FDA methodology for such time series
of complex data objects is still under development.
Our theory covers such series, but the numerical implementation
we developed currently applies only to functions defined on an interval.

This work is motivated both by the need to address a general
inferential problem and by specific data with which we have worked
over the past decade. We first discuss the general motivation, then we
illustrate it using the  data.

 Most inferential procedures for FTS's require that the series be
 stationary (several examples for such procedures can be found in
 \citetext{HKbook}). However, pollution levels, finance or traffic
 data may  exhibit periodic (e.g.\ weekly) patterns,  and then the stationarity
 assumption is violated. \citetext{horvath:kokoszka:rice:2014} propose
 several  testing procedures from the so-called KPSS family to test
 the  stationarity of an FTS. Their approach is based on functionals of
 a CUSUM process, which makes it powerful when testing against changes
 in the mean or against integration of order 1. However, it is not
 designed for testing against a periodic signal.
Finding periodicity in a data set is also of direct relevance
 for understanding the problem at hand as will be illustrated
in Section~\ref{s:realdata}. Exploiting the full
 information contained in the shapes of
functions is crucial. Tests of periodicity
for FTS's can be applied to the observed functions or to
residual functions obtained after fitting some model.
If periodicity is found in the residuals, it may  indicate an
inadequate model fit.

\begin{figure}
\includegraphics[width=7.5cm]{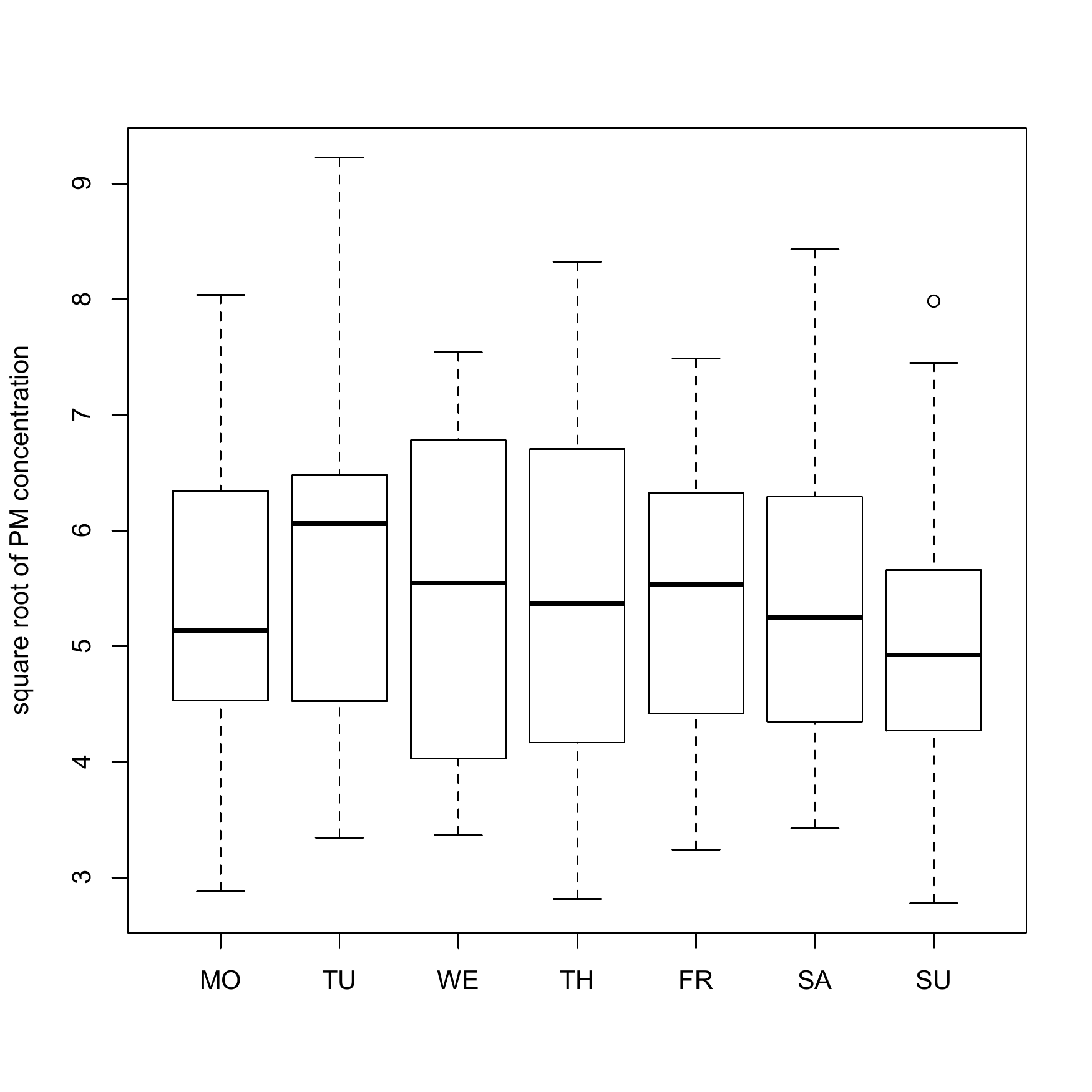}\hfill
\includegraphics[width=7.5cm]{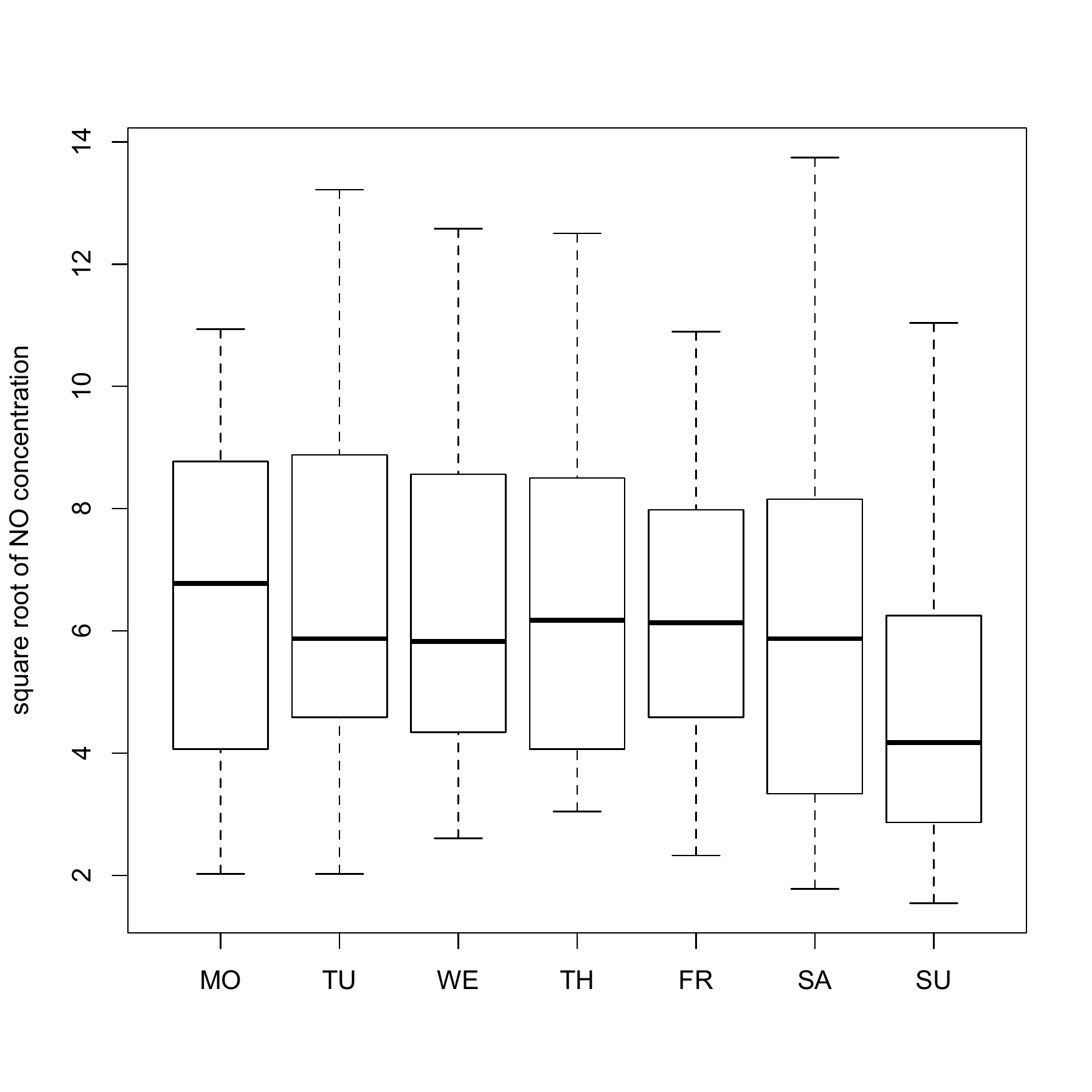}
\caption{Boxplots of {\tt PM10} (left) and {\tt NO} (right)
for each  to day of the week. The sample consists of
 $N$=167 days.\label{fig:boxplots} }
\end{figure}

The following motivating example, which is described in detail in
Section~\ref{s:realdata}, illustrates the need to develop new tests
that exploit the functional structure of the data.
Figure~\ref{fig:boxplots} shows boxplots of daily averages of the
pollutants {\tt PM10} (fine dust) and {\tt NO} (nitrogen monoxide)
measured in Graz, Austria during the winter season 2015/16. The
boxplots are grouped by weekdays and we want to infer if the
corresponding group means differ significantly.  Due to the traffic
exposure of the measuring device in the city center and the weekday
dependent traffic volumes reported in
\citetext{stadlober:pfeiler:2004}, significant differences between the
groups are expected. But although the boxplots indicate lower
concentrations on Sundays, the variation within the groups is
relatively large, and from a one-way ANOVA we do not  obtain evidence
against  the null hypothesis of equal weekday means. The
$p$-values are $0.75$~({\tt PM10}) and $0.27$~({\tt NO}),
respectively.  It needs to be stressed at this point, that formally
ANOVA is not theoretically justified since we are analyzing time
series data which are serially correlated.  Nevertheless, we will see
in Section~\ref{s:realdata} that for both data sets the conclusion
remains the same even after adjusting the test for dependence.  Now
let us look at this problem from a functional data perspective.
Figure~\ref{fig:daymean} shows intraday mean curves (our raw pollution
data are available up to half-hour resolution) of both pollutants
during the same winter season. The plot suggests that Saturday and
Sunday mean curves differ from those of working days. While they have
smaller peaks, they have higher lows (presumably due to lower commuter
traffic and higher nighttime activity on weekends). The methodology
developed in the subsequent sections, will allow us to judge whether the
differences in the functional means are significant. In this
particular example the answer is confirmative. Hence, in contrast to
daily averages, the intraday mean functions do significantly depend on
the day of the week.

\begin{figure}
\begin{center}
\includegraphics[scale=0.4]{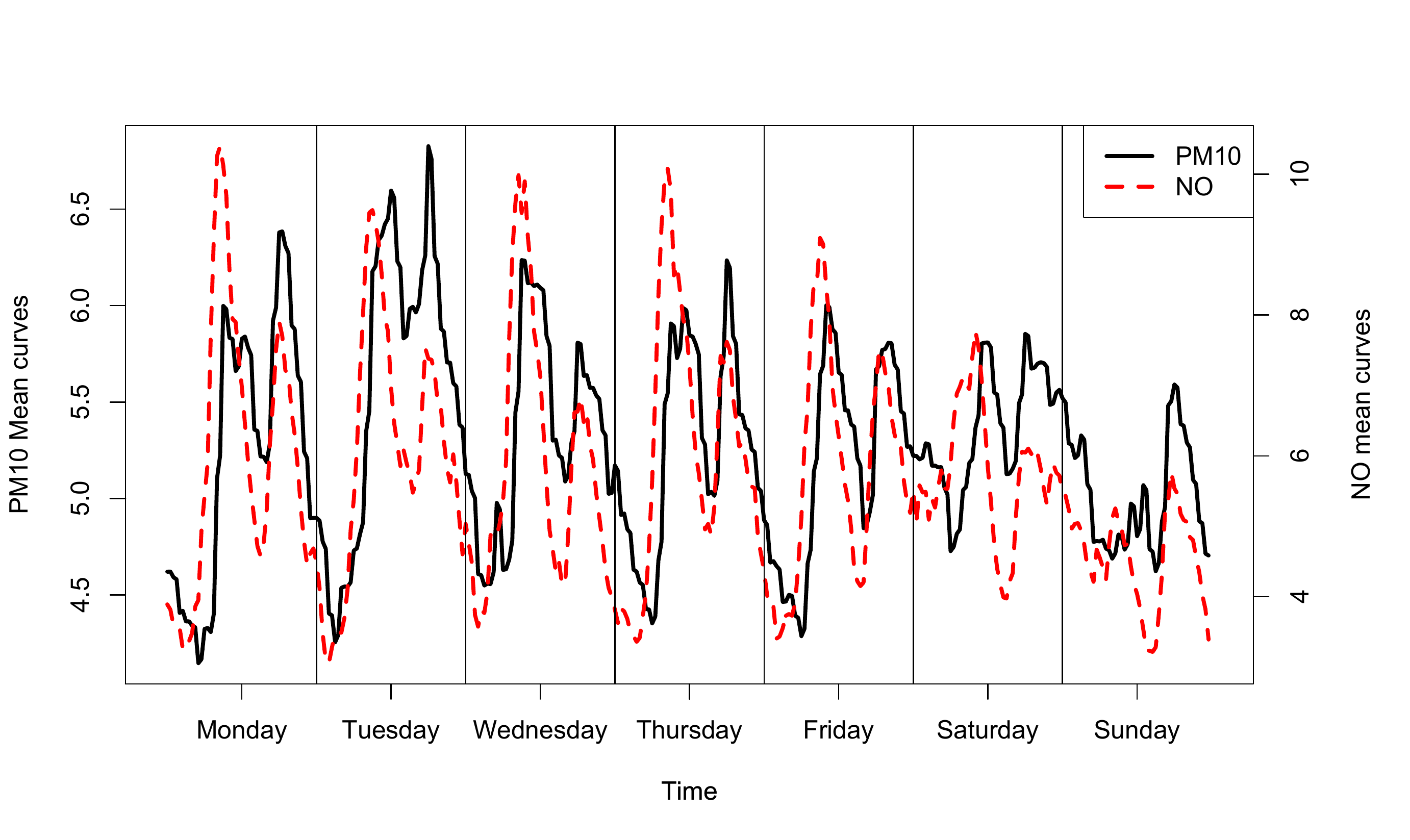}
\end{center}
\caption{Weekday means of {\tt PM10} (solid black) and {\tt NO}
(red dashed). \label{fig:daymean}}
\end{figure}

One of the important contributions of this paper is the development of
a \emph{fully functional ANOVA test for dependent data}.  Using a
frequency domain approach, we obtain the asymptotic null distribution
of the functional ANOVA statistic.  This result is formulated in
Corollary~\ref{c:anova}. The limiting distribution has an interesting
form and can be written as a sum of independent hypoexponential
variables whose parameters are eigenvalues of the spectral density
operator of $(Y_t)$. To the best of our knowledge, there exists no
comparable asymptotic result in FDA literature.

Adapting ANOVA for stationary time series is one way to conduct
periodicity analysis. It is suitable when the periodic component has
no particular form. If, however, the alternative is more specific,  then
we can construct simpler and more powerful tests. In Section~\ref{s:m},
we introduce three different models of  increasing complexity,  and in
Section~\ref{s:tests} we develop the appropriate test statistics. By considering specific local alternatives the
power advantage will be numerically illustrated in
Section~\ref{ss:localpower} and theoretical supported in the
supplemental material (Appendix~\ref{s:local}). General
consistency results are provided in Section~\ref{s:asy}.

We have emphasized so far fully functional testing procedures which
are theoretically elegant and appealing. A common approach to
inference for functional data is to project observations onto a low
dimensional basis system and then to apply a suitable multivariate
procedure to the vector of projections. This approach will be outlined
in Section~\ref{ss:fin}. Our multivariate results improve upon
\citetext{macneill:1974} in two ways: First, our tests are derived
from a (Gaussian) likelihood-ratio approach. As we will see,  this
provides a power advantage over  MacNeill's test. Second, we extend all
our tests in Section~\ref{s:dep} to a very general weak  dependence
setting, as opposed to linear processes studied in
\citetext{macneill:1974} and \citetext{hannan:1961}.

Our methodology and theory for
dependent FTS'  are  based on new developments in the Fourier methods
for such series. The work of Panaretos and Tavakoli
(\citeyear{panaretos:tavakoli:2013AS},
\citeyear{panaretos:tavakoli:2013SPA})
introduces the main concepts of this approach, such as the functional
periodogram and spectral density operators. This framework has been
recently extended and  used in other contexts, see
e.g. \citetext{hormann:kidzinski:hallin:2015} and
\citetext{zhang:2016}. \citetext{zamani:2016} use it in a setting
that falls between our models \eqref{e:test-s} and \eqref{e:test-w}
(iid Gaussian error functions), which also allows them to derive tests
for hidden periodicities; the climate data they study may exhibit some
a priori unspecified periods. For the data that motivate our work
(pollution, traffic, temperature, economic and and finance data), the
potential period is known (week, year, etc.),  and they generally exhibit
dependence under the null. This work therefore focuses on a fixed
known period and weakly dependent  functions.

The remainder of the paper is organized as follows. In Sections
\ref{s:m} and \ref{s:tests}, we consider models and tests under the
null of iid Gaussian functions. Section~\ref{s:dep} considers
dependent non--Gaussian functions. Consistency of the tests is
established in Section~\ref{s:asy}. Applications to pollution data and
a simulation study are presented, respectively, in Sections
\ref{s:realdata} and \ref{s:sim}. In Section~\ref{ss:localpower},  we
numerically assess asymptotic local power of the tests. The main
contributions and findings are summarized in Section~\ref{s:con}.  All
technical results and proofs that are not necessary to understand and
apply the new methodology are presented in the supplemental material.

\section{Models for periodic functional time series} \label{s:m}
The classical model,   \citetext{fisher:1929},
for a (scalar) periodic signal contaminated by noise is
\begin{equation}\label{e:scalar}
y_t=\mu+\alpha \cos(t\theta)+\beta\sin(t\theta) +z_t,
\end{equation}
where the $z_t$ are normal white noise, $\alpha$, $\beta$ and $\mu$ are
unknown constants and $\theta\in [-\pi,\pi]$ is a known frequency which
determines the period.
Model \eqref{e:scalar}  has been extended in several
directions, for example, by replacing a pure harmonic wave by an
arbitrary periodic component and/or by replacing the normal white
noise by a more general stationary time series, as well as by
considering multivariate series.

 In this section, we list extensions
to functional time series organizing them by increasing complexity.
Our theory is valid in an arbitrary separable Hilbert space $H$, in
which $\langle x,y\rangle $ denotes the inner product and
$\|x\|=\sqrt{\langle x,x\rangle}$ the corresponding norm, $x,y\in H$.
In most applications,  it is the space $L^2$ of square integrable
functions on a compact interval, in which case $\langle x,y\rangle=\int
x(u)y(u) du$. A comprehensive  exposition of  Hilbert space theory
for functional data is given in \citetext{hsing:eubank:2015}.

 We begin by stating the following (preliminary)
assumption on the functional noise process.

\begin{Assumption}\label{ass:gauss}
  The noise $(Z_t)$ is an i.i.d.\ sequence in $H$,  with
  each $Z_t$ being a Gaussian
  element in $H$ with zero mean and  covariance operator $\Gamma$.
\end{Assumption}

Recall that a random variable $Z$  in $H$ is
Gaussian, in short $Z\sim \mathcal{N}_{H}(\mu,\Gamma)$, if and only if
all projections $\langle Z,v\rangle$, $v\in H$,  are normally
distributed with mean $\langle \mu, v \rangle$ and variance $\langle
\Gamma(v),v\rangle$.  Working under Assumption~\ref{ass:gauss} is
convenient because we can motivate our tests proposed in
Section~\ref{s:tests} by a likelihood ratio approach and calculate
exact distributions. Nevertheless, this framework is too restrictive
for many applied problems. We devote Section~\ref{s:dep} to procedures
applicable in case of noise which is a general stationary functional
time series.  The testing problems remain the same, but the test
statistics and/or critical values change.

To make the exposition more specific and focused on the main ideas, we
introduce the following assumption.
\begin{Assumption}\label{ass:d}
The sample size $N$ is a multiple of the period, $N=d n$, where the
period $d>1$ is odd. We set $q=(d-1)/2$.
\end{Assumption}
 Appendix~\ref{s:gen-d} discusses modifications needed in case of even
 $d$. Assuming that the sample size $N$ is a multiple of $d$ is not
 really restrictive and can easily be achieved by trimming up to $d-1$
 data points.

The simplest extension of model \eqref{e:scalar} to a functional setting  is
\begin{Mequation} \label{e:model-Z}
Y_t (u) =
\mu(u) + [\alpha \cos(t\theta)+\beta\sin(t\theta)] w(u) +Z_t(u),
\quad \mu,\,
w\in H,\quad \alpha,\beta\in\mathbb{R}.
\end{Mequation}
If $\rho:=\sqrt{\alpha^2+\beta^2}=0$, then $(Y_t\colon t\geq 1)$ is
functional Gaussian white noise with a mean function $\mu$. If $\rho >
0$, then a periodic pattern is added, which varies along the direction
of a function $w$.  To ensure identifiability, we assume that
$\|w\|^2:=\int_0^1w^2(u)du=1$. The functions $\mu$ and $w$, as well as
the parameters $\alpha$ and $\beta$ are assumed to be unknown.
As explained in the Introduction,
the parameter $\theta$, which determines the period $d$, is
assumed to be a {\em known} positive fundamental frequency, i.e
\[
\theta \in \Theta_{N}:=\left \{\theta_{j}=2\pi j/N, \;
j=1, \ldots, m:=\left[(N-1)/2\right]\right \}.
\]
 The testing problem is
\begin{equation} \label{e:tp}
\mathcal{H}_0\colon \rho=0
\ \ \ {\rm vs.} \ \ \
\mathcal{H}_A\colon \rho> 0.
\end{equation}

A first extension of \eqref{e:model-Z} is to replace
$\alpha\cos(\theta t)+\beta\sin(\theta t)$ by some arbitrary
$d$--periodic sequence.  A more general model thus is
\begin{Mequation}\label{e:model-s}
Y_t(u)=\mu(u)+s_t w(u)+Z_t(u),\quad s_t=s_{t+d}, \ \ \ \mu, w\in H.
\end{Mequation}
We wish to test
\begin{equation}\label{e:test-s}
\mathcal{H}_0\colon s_1=s_2=\cdots=s_d=0
\quad\text{against}\quad
\mathcal{H}_A\colon \max_{1\leq t\leq d}|s_t| >  0.
\end{equation}
Here we impose the identifiability constraints $\|w\|=1$ and $\sum_{k=1}^d s_t=0$. The latter ensures that
the vector $(s_1,\ldots,s_d)^\prime$ is contained in the subspace
spanned by the orthogonal vectors
\[
\begin{pmatrix}
\cos(\theta_n)\\ \cos(2\theta_n)\\ \vdots\\ \cos(d\theta_n)
\end{pmatrix},
\begin{pmatrix}
\sin(\theta_n)\\ \sin(2\theta_n)\\ \vdots\\ \sin(d\theta_n)
\end{pmatrix},\ldots,
\begin{pmatrix}
\cos(\theta_{nq})\\ \cos(2\theta_{nq})\\ \vdots\\ \cos(d\theta_{nq})
\end{pmatrix},
\begin{pmatrix}
\sin(\theta_{nq})\\ \sin(2\theta_{nq})\\ \vdots\\ \sin(d\theta_{nq})
\end{pmatrix},
\]
cf. Assumption~\ref{ass:d}. With the convention
\begin{equation}\label{eq:vartheta}
\vartheta_k:=\theta_{nk}=2\pi k/d,
\end{equation}
model \eqref{e:model-s} can be
written as
\begin{equation} \label{e:rep-s}
Y_{t}(u)=\mu(u) +
\left(\sum_{k=1}^{q} \big(\alpha_{k} \cos(t \vartheta_k)
+ \beta_{k} \sin(t \vartheta_k)\big)\right)w(u)+Z_{t}(u),
\end{equation}
with some coefficients $\alpha_k$ and $\beta_k$.

Model \eqref{e:model-s} assumes that at any point of time, the
periodic functional component is proportional to a single function
$w$. A model which imposes periodicity in a very general
sense is
\begin{Mequation}\label{e:model-G}
Y_t(u)=\mu(u)+w_t(u)+Z_t(u),\quad \mu,w_t\in H,\quad w_t=w_{t+d},
\quad \sum_{t=1}^d w_t=0.
\end{Mequation}
In this context, we  test
\begin{equation}\label{e:test-w}
\mathcal{H}_0\colon w_1=w_2=\cdots=w_d=0\quad\text{against}\quad
\mathcal{H}_A\colon \max_{1\leq t\leq d} \|w_t\| > 0.
\end{equation}
Model \eqref{e:model-G} contains models \eqref{e:model-Z} and
\eqref{e:model-s} as special cases. Under $\mathcal{H}_0$ all three
models are identical. Test procedures presented in Section~\ref{s:tests}
are motivated by specific models as they point toward specific alternatives.
However, they can be applied to any data, and, as we demostrate in
Sections \ref{s:realdata} and \ref{s:sim},
tests motivated by simple models often perform very well for more
complex alternatives.

\section{Test procedures in presence of  Gaussian noise}
\label{s:tests}
In the following subsections, we present the basic form of periodicity
tests. Throughout this section, we work under
Assumptions~\ref{ass:gauss} and \ref{ass:d}. In Section~\ref{s:dep}
and Appendix~\ref{s:gen-d}, respectively, we show how to remove these
assumptions. Details of mathematical derivations are given in
Appendix~\ref{s:ta}.

Let us start by introducing the necessary notation and notational
conventions. Given a vector time series $(\boldsymbol{Y}_t\colon 1\leq
t\leq N)$ the discrete Fourier transform (DFT) is
$\boldsymbol{D}(\theta)=\frac{1}{\sqrt{N}}\sum_{k=1}^N
\boldsymbol{Y}_k e^{-\mathrm{i}k\theta}$, $\theta\in[-\pi,\pi]$. We
will use  the decomposition into real and complex parts:
$\boldsymbol{D}(\theta)=\boldsymbol{R}(\theta)+\mathrm{i}\boldsymbol{C}(\theta)$. At
some places we may add a subscript to indicate the dependence on the
sample size and/or a superscript to refer to the underlying
data. (E.g.\ $\boldsymbol{R}_N^{\boldsymbol{Y}}(\theta)$.) We proceed
analogously for a functional time series $(Y_t\colon 1\leq t\leq
N)$. Then the DFT is denoted by
$D(\theta)=R(\theta)+\mathrm{i}C(\theta)$.

Let us set
$
\boldsymbol{A}(\theta_{i_1},\ldots,\theta_{i_k})=\big[\boldsymbol{R}(\theta_{i_1}),\ldots,\boldsymbol{R}(\theta_{i_k}),\boldsymbol{C}(\theta_{i_1}),\ldots,\boldsymbol{C}(\theta_{i_k})\big]^\prime
$ and analogously $A(\theta_{i_1},\ldots,\theta_{i_k})$ be a
$2k$-vector of functions with components $R(\theta_{i_j})$ and
$C(\theta_{i_j})$. If $A=(A_1,\ldots,A_k)^\prime$ is any $k$-vector of
functions, then $A A^\prime$ is the $k\times k$ matrix of scalar
products $\langle A_i,A_j\rangle$. We use $\|M\|$ for the usual
(Euclidean) norm and $\|M\|_{\mathrm{tr}}$ for the trace norm of some
generic matrix $M$. Finally, $\bW_p(n)$ denotes the real $p\times p$
Wishart matrix with $n$ degrees of freedom and $q_{\alpha}(X)$ is the
$\alpha$-quantile of some variable $X$.

\subsection{Projection based approaches}
\label{ss:fin}
Typically functional data are represented in a smoothed form by finite
dimensional systems, such as B--splines, Fourier basis, wavelets,
etc. Additional dimension reduction can be achieved by functional
principal components or similar data--driven systems.
It is thus natural to search for a periodic
pattern within a lower dimensional approximation of the data.

In this section, we assume that $v_1, v_2, \ldots, v_p$ is a suitably
chosen set of linearly independent functions.  Setting
$\boldsymbol{Y}_t:=(\langle Y_t,v_1\rangle,\ldots,\langle
Y_t,v_p\rangle)^\prime$, we obtain vector observations. Under
$\mathcal{H}_0$,  the time series $(\boldsymbol{Y}_t)$ is i.i.d.\
Gaussian with covariance matrix
$\boldsymbol{\Sigma}=(\langle\Gamma(v_i),v_j\rangle\colon 1\leq i,
j\leq p)$.  Under $\mathcal{H}_A$ we can write the projected version
of model~\eqref{e:model-G} as
\begin{equation}\label{e:vector}
\boldsymbol{Y}_t=\boldsymbol{\mu}+\boldsymbol{w}_t
+\boldsymbol{Z}_t,
\end{equation}
with $\boldsymbol{\mu}=(\langle \mu,v_1\rangle\cdots \langle
\mu,v_p\rangle)'$, $\boldsymbol{w}_t=(\langle w_t,v_1\rangle\cdots
\langle w_t,v_p\rangle)'$ and the innovations $\boldsymbol{Z}_t=(\langle
Z_t,v_1\rangle\cdots \langle Z_t,v_p\rangle)'$. This in turn can be
specialized to projected versions of models~\eqref{e:model-Z} and
\eqref{e:model-s}. Provided the periodic component in the investigated
model is not orthogonal to \ $\mathrm{span}\{v_1, v_2, \ldots, v_p\}$,
we can formulate the corresponding multivariate testing problems. In
the following theorem we state   the  likelihood ratio
tests.  Recall the definition of the frequencies $\vartheta_k$ in
\eqref{eq:vartheta} and the notation $q=(d-1)/2$.

\begin{Theorem}\label{th:LR}
For a given positive definite $\boldsymbol{\Sigma}$,
 the likelihood-ratio tests for the
multivariate analogues of testing problems \eqref{e:tp},
\eqref{e:test-s} and \eqref{e:test-w} (related to the projected models
\eqref{e:model-Z}, \eqref{e:model-s} and \eqref{e:model-G},
respectively) are given as follows: Reject the null-hypothesis at
level $\alpha$ if
\begin{align*}
T^{{\tt MEV}_1}&:=\big\|\boldsymbol{A}(\vartheta_1)\boldsymbol{\Sigma}^{-1}\boldsymbol{A}^\prime(\vartheta_1)\big\|> q_{1-\ag}[\|\bW_p(2)\|/2];\\
T^{{\tt MEV}_2}&:=\big\|\boldsymbol{A}(\vartheta_1,\ldots,\vartheta_q)\boldsymbol{\Sigma}^{-1}\boldsymbol{A}^\prime(\vartheta_1,\ldots,\vartheta_q)\big\| > q_{1-\ag}[\|\bW_p(d-1)\|/2];\\
T^{{\tt MTR_2}}&:=\big\|\boldsymbol{A}(\vartheta_1,\ldots,\vartheta_q)\boldsymbol{\Sigma}^{-1}\boldsymbol{A}^\prime(\vartheta_1,\ldots,\vartheta_q)\big\|_\mathrm{tr}>q_{1-\ag}[ \mathrm{Erlang}(pq,1)].
\end{align*}
\end{Theorem}

Some remarks are in order.
\begin{enumerate}
\item The superscript {\tt MEV} in our tests stands for {\tt M}ultivariate {\tt E}igen{\tt V}alue. \emph{Multivariate}, as opposed to functional, and \emph{eigenvalue},  refers to the fact that the Euclidean matrix norm
of a symmetric matrix is equal to its largest eigenvalue.
{\tt MTR} abbreviates {\tt M}ultivariate {\tt TR}ace.
\item By
Lemma~\ref{le:orth}, the columns of
$\boldsymbol{\Sigma}^{-1/2}\boldsymbol{A}^\prime(\vartheta_1,\ldots,\vartheta_q)$
are i.i.d.\ $\mathcal{N}_p(0,\frac{1}{2}I_p)$. This explains the
Wishart distribution. For
explicit computation of the quantiles $q_{1-\ag}[\|\bW_p(k)\|]$ we
refer to \citetext{chiani:2014}.
\item An alternative to the test based on $T^{{\tt MEV}_1}$ is
\[
T^{{\tt MTR}_1}=\big\|\boldsymbol{A}(\vartheta_1)\boldsymbol{\Sigma}^{-1}\boldsymbol{A}^\prime(\vartheta_1)\big\|_\mathrm{tr}
> q_{1-\ag}[ \mathrm{Erlang}(p,1)].
\]
The latter can be seen to be equivalent to the test proposed by
\citetext{macneill:1974} for a multivariate version of model
\eqref{e:model-Z}. The likelihood ratio and MacNeill's test statistic
are related to different matrix norms of
$\boldsymbol{A}(\vartheta_1)\boldsymbol{\Sigma}^{-1}\boldsymbol{A}^\prime(\vartheta_1)$. By
the Neyman--Pearson lemma, a likelihood ratio test, even in an
approximate form, can be expected to have good and sometimes even
optimal power properties.  Likewise, replacing the matrix norm in $T^{{\tt
MEV}_2}$ by the trace norm leads to $T^{{\tt MTR_2}}$. As
Figure~\ref{fig:powerq2} illustrates, the difference in power between
the two tests can be quite noticeable, especially when
$d$ is large.
\item In practice, $\boldsymbol{\Sigma}$ must  be replaced
by a consistent estimator. The construction of such estimators, which remain consistent under $\mathcal{H}_A$,
is discussed in Appendix~\ref{s:prac}.
\end{enumerate}

\begin{figure}
\includegraphics[width=7cm]{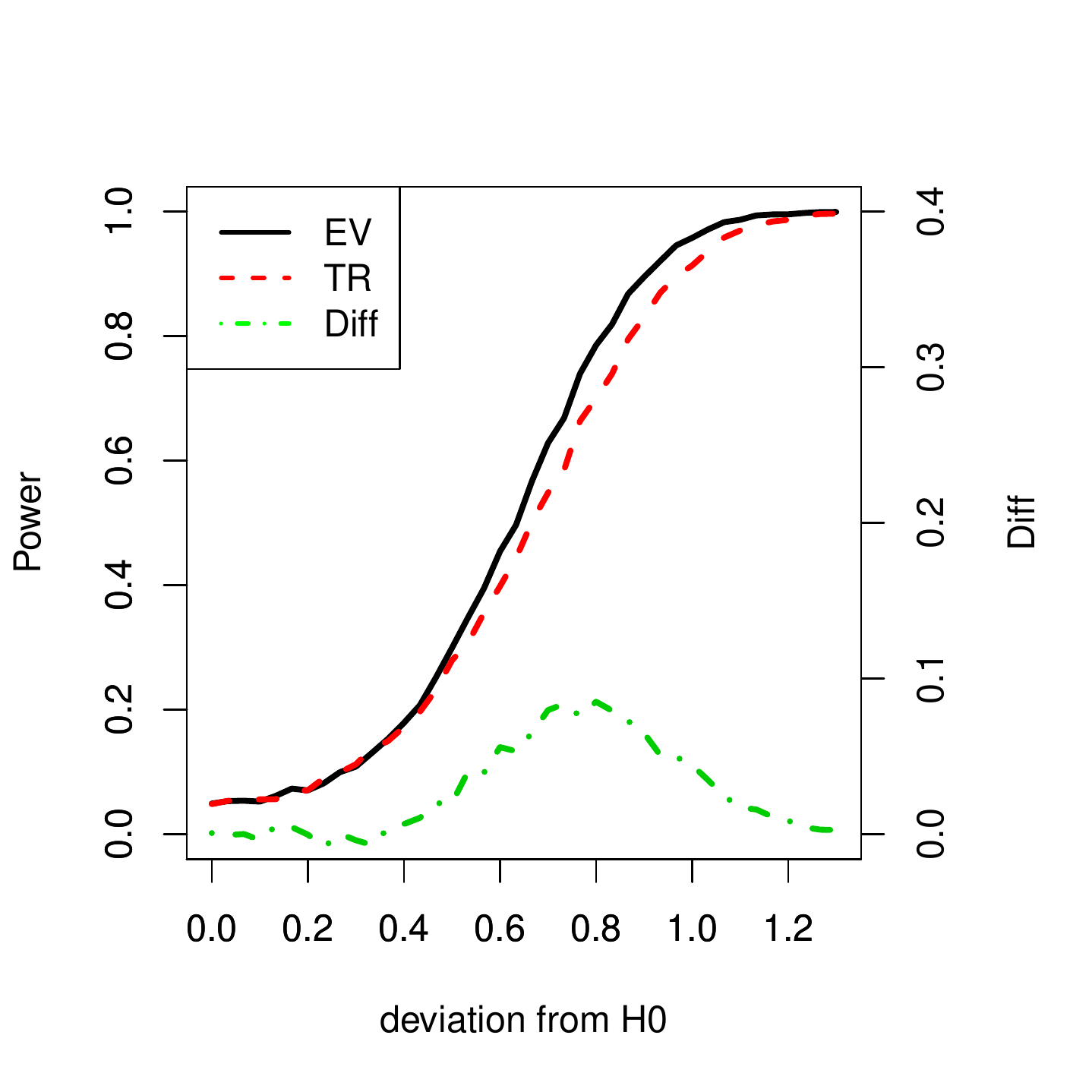}
\includegraphics[width=7cm]{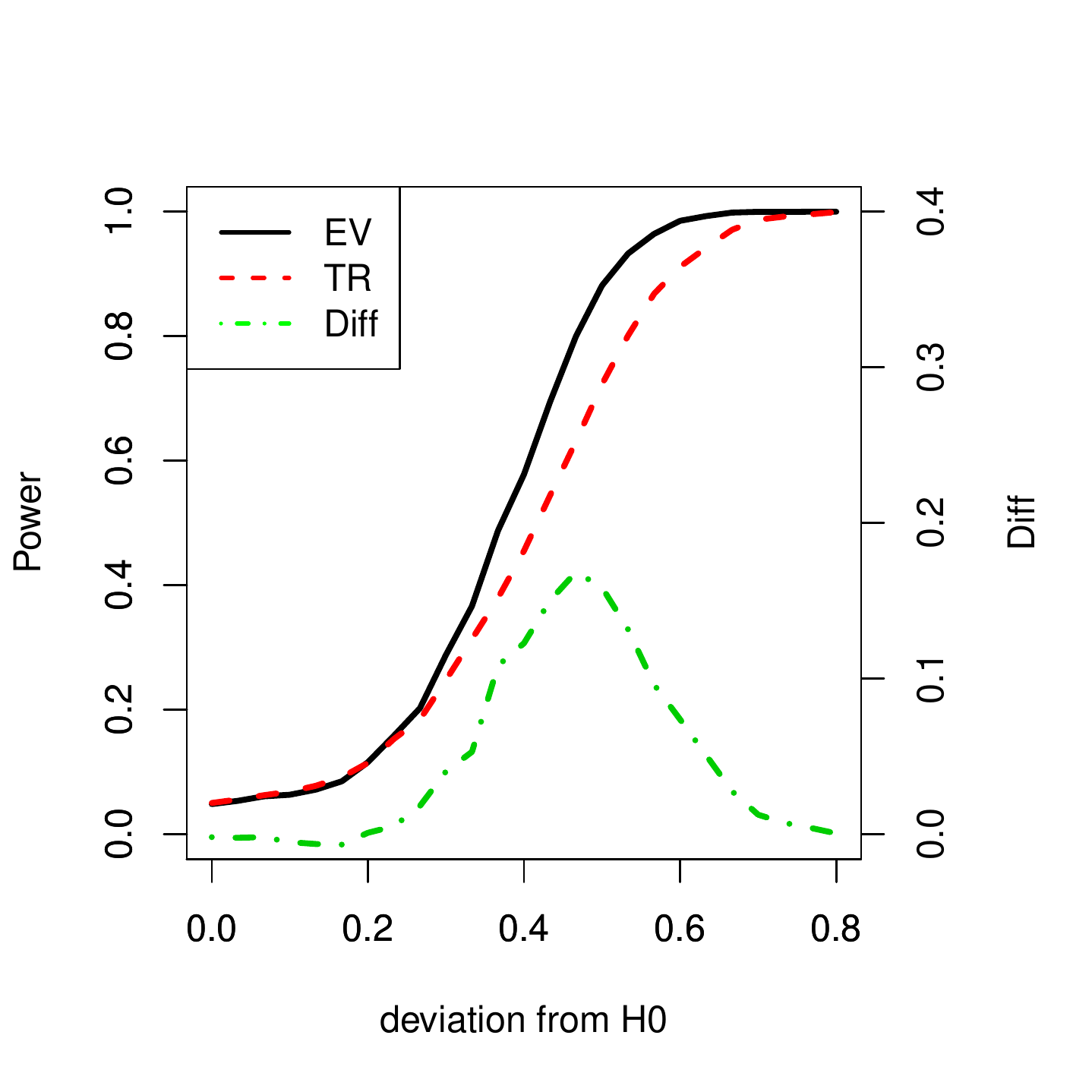}
\caption{Local asymptotic power curves of tests $T^{{\tt MEV}_2}$ (EV)
and $T^{{\tt MTR_2}}$ (TR), and their  difference (Diff, right scale).
 Left panel:  $p=5$ and $d=7$, right panel: $p=5$ and $d=31$. Details of the implementation are given in Section~\ref{ss:localpower}.
\label{fig:powerq2}}
\end{figure}

\subsection{Fully functional tests}
\label{ss:ff}
The projection based approaches of the previous
section may be  sensitive to the choice of the basis and to the number of
basis functions. It is hence desirable to develop some fully
functional procedures to bypass this problem. Before we introduce
fully functional test statistics, let us observe that $T^{\tt MEV_i}$
and $T^{\tt MTR_i}$ $(i=1,2)$ are computed from the rescaled sample
$\bSig^{-1/2}\boldsymbol{Y}_1,\ldots,\bSig^{-1/2}\boldsymbol{Y}_N$,
which results in asymptotically pivotal tests. The rescaling
guarantees that the component processes with larger variation are not
concealing potential periodic patterns in components with little
variance. While this is clearly a very desirable property in
multivariate analysis, one may favor a different perspective for
functional data.  If $\boldsymbol{Y}_t$ are principal component
scores, then $\bSig=\mathrm{diag}(\lambda_1,\ldots,\lambda_p)$, where
$\lambda_i$ are the eigenvalues of $\mathrm{Cov}(Z_1)$. Suppose that
$Y_t(u)=\sqrt{\lambda_\ell}\cos(2\pi t/d)v_\ell(u)+Z_t(u)$, $\ell\geq 1$. Then, due
to $\lambda_\ell\to 0$, the bigger $\ell$, the smaller and more
negligible the periodic signal is. However, it is easily seen that for
any of our multivariate tests,  the probability of rejecting
$\mathcal{H}_0$ is the same for all values $1\leq \ell\leq p$.

A way to account for the functional nature of the data is  to base
the test statistics directly on the unscaled and fully functional observations, i.e.\ to define analogues of the test statistics in Theorem~\ref{th:LR}
with the matrices $A(\vartheta_1)A^\prime(\vartheta_1)$ (in $\mathbb{R}^{2\times 2}$) and $A(\vartheta_1,\ldots,\vartheta_q)A^\prime(\vartheta_1,\ldots,\vartheta_q)$ (in $\mathbb{R}^{(d-1)\times (d-1)}$). Since, to the best of our knowledge, there is no result available on the distribution of $\big\|A(\vartheta_1,\ldots,\vartheta_q)A^\prime(\vartheta_1,\ldots,\vartheta_q)\big\|$, we shall only consider the trace norm for which we can get explicit formulas. Hence, for model~\eqref{e:model-Z} we propose a test which
rejects $\mathcal{H}_{0}$ at level $\alpha$ if
$$
T^{{\tt FTR_1}}:=\|A(\vartheta_1)A^\prime(\vartheta_1)\|_\mathrm{tr}
>q_{1-\ag}
[\mathrm{HExp}(\lambda_{1},\lambda_{2},\ldots)].
$$
Here $\mathrm{HExp}(\lambda_{1},\lambda_{2},\ldots )$ denotes a random
variable which is distributed as $\sum_{i\geq 1} \lambda_{i} E_{i}$,
where the $E_{i} $ are i.i.d.\ $\mathrm{Exp}(1)$ variables. If
$\lambda_i=0$ for $i>k$,  then this is a so-called hypoexponential
distribution, whose distribution function is explicitly known, see
e.g.\ \citetext{ross:2010}, Section~5.2.4.
For models~\eqref{e:model-s} and \eqref{e:model-G} we propose the test which rejects $\mathcal{H}_{0}$ at level $\alpha$ if
\begin{equation}\label{e:tr-fq}
T^{{\tt FTR_2}}:=\|A(\vartheta_1,\ldots,\vartheta_q)A^\prime(\vartheta_1,\ldots,\vartheta_q)\|_\mathrm{tr}>q_{1-\ag}
\left [ \sum_{k=1}^q\Xi_k \right ],
\end{equation}
where $\Xi_k\stackrel{\text{i.i.d.}}{\sim}
\mathrm{HExp}(\lambda_1,\lambda_2,\ldots)$.
Lemma~\ref{le:Hexp} provides the justification
of \eqref{e:tr-fq}.

In practice we will approximate $\mathrm{HExp}(\lambda_1,\lambda_2,\ldots)$ by $\mathrm{HExp}(\hat\lambda_1,\hat\lambda_2,\ldots,\hat\lambda_k)$ with eigenvalues $\hat\lambda_i$ of $\widehat\Gamma$ and some fixed
$k$ to obtain critical values. (See Section~\ref{s:prac}.) Since the sample covariance has only a finite number of non-zero eigenvalues, we can either use all of them or  chose the smallest $k\geq 1$ such that $\mathrm{tr}(\widehat \Gamma)-\big(\hat\lambda_1+\cdots+\hat\lambda_k\big)\leq \epsilon$ for some small $\epsilon$.
Other details are presented in Appendix~\ref{s:prac}.

\subsection{Relation to MANOVA and functional ANOVA} \label{ss:any}
A possible strategy for our testing problem is to embed it into the
ANOVA framework as it was sketched in the Introduction. If the period
is $d$, we can think of the data as coming from $d$ groups, and the
objective is to test if all groups have the same mean.  ANOVA can be
applied to  models \eqref{e:model-Z} and \eqref{e:model-s}, but it is
particularly suitable for model \eqref{e:model-G} where we impose no
structural assumptions on the periodic component.  As in the previous
sections, we can either adopt a multivariate setting, where we
consider projections onto specific directions, or a fully functional
approach.

The likelihood ratio  statistic in the multivariate setting is the
classical MANOVA test based on Wilk's lambda (see~\citetext{mardia:kent:bibby:1980}) which
 is given as the  ratio of the determinants of the empirical
covariance under $\mathcal{H}_0$ in the numerator and of the empirical
covariance under $\mathcal{H}_A$ in the denominator. Such an object is
not easy to extend to the fully functional setting. If, however, we
work with a fixed $\bSig$ (later it can
be replaced by an estimator), then the LR statistic takes the  form
\begin{equation} \label{e:MAV}
T^{\tt MAV}=\frac{1}{d}\sum_{k=1}^d n(\overline{\boldsymbol{Y}}_k-\overline{\boldsymbol{Y}})^\prime \boldsymbol{\Sigma}^{-1}(\overline{\boldsymbol{Y}}_k-\overline{\boldsymbol{Y}}),
\end{equation}
where
$\overline{\boldsymbol{Y}}_k=\frac{1}{n}\sum_{t=1}^n\boldsymbol{Y}_{(t-1)d+k},
1\leq k\leq d$,  and $\overline{\boldsymbol{Y}}$ is the grand mean.
Translating this, with the same line of argumentation as in Section~\ref{ss:ff}, into the fully functional setting we obtain
\begin{equation} \label{e:TAV}
T^{{\tt FAV}}
=\frac{1}{d}\sum_{k=1}^d n\|\overline{Y}_k-\overline{Y}\|^2,
\end{equation}
where $\overline{Y}_k$ and $\overline{Y}$ are defined analogously. This formally coincides with the functional ANOVA test statistics
considered in \citetext{cuevas:febrero:freiman:2004} assuming a balanced design.

The following important result shows that the test statistics
 \eqref{e:MAV} and
\eqref{e:TAV} are equivalent to $T^{\tt MTR_2}$ and $T^{\tt FTR_2}$,
respectively.

\begin{Proposition}\label{pr:anova_tr}
It holds that
$T^{\tt{MAV}}
=\frac{2}{d}T^{\tt MTR_2}\quad\text{and}\quad T^{\tt{FAV}}
=\frac{2}{d}T^{\tt FTR_2}.$
\end{Proposition}
Proposition~\ref{pr:anova_tr} is proven in Appendix~\ref{s:ta}.  We
stress that the equalities in this result are of an algebraic nature,
so they hold for any process $(Y_t\colon t\in\mathbb{Z})$. The
limiting distribution of $T^{\tt FTR_2}$ with general stationary noise
will follow from the  theory developed in Section~\ref{s:dep}.
Hence, we obtain the asymptotic null distribution of the functional
ANOVA statistics $T^{\tt{FAV}}$ for stationary FTS. This is formulated
as Corollary~\ref{c:anova}. The result is of independent interest,  as
it relaxes the independence assumption in the functional ANOVA
methodology.

\section{Dependent non--Gaussian noise}\label{s:dep}
In this section, we derive extensions of the testing procedures
proposed in Section~\ref{s:tests} to the setting of a general
stationary noise sequence $(Z_t)$; we drop the assumptions of
Gaussianity and independence.  We require that $(Z_t)$ be a mean zero
stationary sequence in $H$ which satisfies the following dependence
assumption.

\begin{Assumption}[$L^r$--$m$--approximability]\label{a:approx}
The sequence $(Z_t\colon t\in\mathbb{Z})$ can be represented as
$
Z_t = f(\delta_t,\delta_{t-1},\delta_{t-2}, \ldots),
$
where the $\delta_i$'s are i.i.d.\ elements taking values in some measurable space $S$ and $f$ is a measurable function $f : S^\infty \rightarrow H$. Moreover, if $\delta'_1,\delta'_2,...$ are independent copies of $\delta_1,\delta_2,...$ defined on the same measurable space $S$, then, for
\begin{align*}
Z_t^{(m)}:= f(\delta_t,\delta_{t-1},\delta_{t-2}, \ldots ,
\delta_{t-m+1},\delta'_{t-m},\delta'_{t-m-1}, \ldots ),
\end{align*}
we have
\begin{align}\label{lpm}
\sum_{m=1}^\infty (E\|Z_m - Z_m^{(m)}\|^r)^{1/r} < \infty.
\end{align}
\end{Assumption}

In the context of functional time series, the above assumption was
introduced by \citetext{hormann:kokoszka:2010}, and used in many
subsequent papers including \citetext{hormann:horvath:reeder:2013},
\citetext{horvath:kokoszka:rice:2014}, \citetext{zhang:2016}, among many
others. Similar conditions were used earlier by \citetext{wu:2005} and
\citetext{shao:wu:2007}, to name representative publications.
In the following,  we will use this assumption with $r=2$.
The asymptotic theory could most likely be developed under different
weak dependence assumptions.  The advantage of using
Assumption~\ref{a:approx} is that it has been verified for many
functional time series models, and a number of asymptotic results
exist,  which we can use as components of the proofs.

Denote by $C_h=E (Z_h \otimes Z_0) $ the lag $h$ autocovariance
operator.  If $H$ is the space of square integrable functions, $C_h$
is a kernel operator, $C_h:L^2\to L^2$, which maps a function $f$ to
the function $ C_h(f)(u)=\int E[Z_h(u)Z_0(s)]f(s)ds$.  If
Assumption~\ref{a:approx} holds with $r=2$, then
\begin{equation}\label{e:abssymcov}
\sum_{h\in\mathbb{Z}}\|C_h\|_\mathcal{S}<\infty,
\end{equation}
where $\|\cdot\|_\mathcal{S}$ denotes the Hilbert-Schmidt norm.  As
shown in \citetext{hormann:kidzinski:hallin:2015}, this ensures the
existence of the  \emph{spectral density operator}:
\[
\mathcal{F}_\theta:=\sum_{h\in\mathbb{Z}}
C_h e^{-\mathrm{i} h\theta}.
\]
This operator was  defined in
\citetext{panaretos:tavakoli:2013AS} (with an additional scaling factor $\frac{1}{2\pi}$). It
 plays a crucial role in frequency domain analysis of functional time
 series.  We will see in Theorem~\ref{th:asympt} below that the
 spectral density operator is the asymptotic covariance operator of
 the discrete Fourier transform $D_N^Z(\theta)$, and hence it will enter
 the construction of our test statistics in a similar way as
 $\Gamma=\mathrm{Var}(Z_1)$ does in the case of independent noise. We recall hereby the definition of a complex-valued functional Gaussian random variable with mean $\mu$, variance operator $F(v)=E\big[ (X-\mu) \langle v,X-\mu\rangle\big]$ and relation operator $C(v)=E\big[ (X-\mu) \langle v,\overline{X-\mu}\rangle\big]$.  Then $Z=Z_{0} + \mathrm{i} Z_{1}$ with $Z_{0},Z_{1} \in H$ is complex Gaussian $\mathcal{N}_H(\mu, F,C)$ if
 \[
 \begin{pmatrix} Z_{0} \\ Z_{1} \end{pmatrix} \sim \mathcal{N}_{H\times H} \begin{pmatrix} \begin{pmatrix} \mu_{0} \\ \mu_{1} \end{pmatrix}, \frac{1}{2} \begin{pmatrix} \mathrm{Re}(F+C) & -\mathrm{Im}(F-C) \\ \mathrm{Im}(F+C) & \mathrm{Re}(F-C) \end{pmatrix} \end{pmatrix},
 \]
 where $\mu=\mu_{0}+i\mu_{1}=E(Z_{0}+iZ_{1})=\mu$.
 When the relation operator is null, we will write $Z\sim \mathcal{CN}_{H}\left(0,F\right)$.  Theorem~\ref{th:asympt} follows from Theorem~5 in
 \citetext{cerovecki:hormann:2016}.

 \begin{Theorem}\label{th:asympt} If  $(Z_t)$ is  an
   $L^2-m$--approximable time series with values in a separable Hilbert space $H$,  then for any $\theta\in[-\pi,\pi]$
\begin{equation*}
D_N^Z(\theta) \overset{d}\longrightarrow
\mathcal{CN}_{H}\left(0, \mathcal{F}_{\theta}\right).
\end{equation*}
Furthermore,
\begin{enumerate}[(i)]
\item $\mathrm{Var}(D_N^Z(\theta))$ converges in weak
operator topology to $\mathcal{F}_{\theta}$.
\item The components of $(D_N^Z(\theta),D_N^Z(\theta^\prime))$
are asymptotically independent whenever $\theta+ \theta^\prime\neq 0$
and $\theta-\theta^\prime\neq 0$.
\end{enumerate}
\end{Theorem}

Using Theorem~\ref{th:asympt}, which is applicable to both functional and
multivariate data, we are now ready to explain how to construct
tests  when
Assumption~\ref{ass:gauss} is dropped and replaced by
$L^2-m$--approximability. These tests,
justified in Appendix~\ref{s:asy-p}, have asymptotic (rather than exact)
size $\alpha$.

\medskip

\noindent
{\bf Independent noise:} The tests of Section~\ref{s:tests}
remain  unchanged
for general i.i.d.\ noise with second order moments.\medskip

\noindent
{\bf Projection based approach:} If we project the data onto a
 basis $(v_{1},\ldots,v_{p})$, then the resulting multivariate time
 series $\boldsymbol{Y}_{t}$ inherits $L^2$-$m$-approximability. Let
 $\boldsymbol{\mathcal{F}}_\theta$ denote the spectral density matrix of this
 process.  Assuming that the $\boldsymbol{\mathcal{F}}_{\vartheta_{j}}$  are of full rank, we need to replace the matrix $$\boldsymbol{A}(\vartheta_1,\ldots,\vartheta_\ell)\boldsymbol{\Sigma}^{-1}\boldsymbol{A}^\prime(\vartheta_1,\ldots,\vartheta_\ell),\quad \text{$\ell=1$ or $\ell=q$}$$ in the definition of the multivariate tests by
 $$\boldsymbol{H}(\vartheta_1,\ldots,\vartheta_\ell)\boldsymbol{H}^\prime(\vartheta_1,\ldots,\vartheta_\ell),$$ where the columns of $\boldsymbol{H}^\prime(\vartheta_1,\ldots,\vartheta_\ell)$ are given by
 $$
 \Big[\mathrm{Re}\big(
\boldsymbol{\mathcal{F}}_{\vartheta_1}^{-1/2} \boldsymbol{D}(\vartheta_1)
 ,\ldots,
\boldsymbol{\mathcal{F}}_{\vartheta_\ell}^{-1/2} \boldsymbol{D}(\vartheta_\ell)\big),
\mathrm{Im}\big(
\boldsymbol{\mathcal{F}}_{\vartheta_1}^{-1/2} \boldsymbol{D}(\vartheta_1)
,\ldots,
\boldsymbol{\mathcal{F}}_{\vartheta_\ell}^{-1/2} \boldsymbol{D}(\vartheta_\ell)
\big)\Big].
$$
The critical values remain the same as
in Section~\ref{s:tests}.\medskip

\noindent
{\bf Fully functional approach:} In contrast to the multivariate setting the fully functional test statistics remain unchanged, but the critical values need to be adapted according to the following result.
\begin{Proposition}\label{pr:assfundep} If $(Y_t)$ is $L^2$--$m$--approximable then for any frequencies $0<\omega_1<\omega_2<\cdots<\omega_\ell<\pi$,
\[
\|A(\omega_1,\ldots,\omega_{\ell})A^\prime(\omega_1,\ldots,\omega_{\ell})\|_\mathrm{tr}
\stackrel{d}{\to} \sum_{k=1}^\ell\Xi_k,
\] where $\Xi_k\stackrel{\text{ind.}}{\sim}
\mathrm{HExp}(\lambda_1(\omega_k),\lambda_2(\omega_k),\ldots)$,
 and $\lambda_\ell(\omega_{k})$ are the eigenvalues of
 $\mathcal{F}_{\omega_{k}}$.
\end{Proposition}

In practice, we don't know the spectral densities which are necessary
for our tests. In Appendix~\ref{s:prac}, we show how to construct
their  estimators.

We conclude this section with a corollary to Proposition
\ref{pr:assfundep}. This result
 is new and interesting in itself.
It broadly extends the applicability of functional ANOVA by revealing its asymptotic distribution when the underlying data are weakly dependent.
\begin{Corollary} \label{c:anova}
Under the assumptions
of Proposition~\ref{pr:assfundep} the functional ANOVA test statistic satisfies
\[\frac{1}{d}\sum_{k=1}^d n\|\overline{Y}_k-\overline{Y}\|^2\stackrel{d}{\to}
\frac{2}{d}\sum_{k=1}^q \Xi_k,
\]
where $\Xi_k\stackrel{\text{ind.}}{\sim}
\mathrm{HExp}(\lambda_1(\vartheta_k),\lambda_2(\vartheta_k),\ldots)$,
 and $\lambda_\ell(\vartheta_{k})$ are the eigenvalues of
 $\mathcal{F}_{\vartheta_{k}}$.
\end{Corollary}

\section{Consistency of the tests} \label{s:asy} In this section, we
state consistency results for the tests developed in the previous
sections. The proofs are presented in Appendix~\ref{s:asy-p}.  We
focus on the general model \eqref{e:model-G} with the noise $(Z_t)$
satisfying Assumption~\ref{a:approx} with $r=2$, but we also consider
the simpler tests and alternatives introduced in Section~\ref{s:m}. We
assume throughout that Assumption~\ref{ass:d} holds.

To state the consistency results,
we decompose the DFT of the functional
observations as follows:
\begin{equation}\label{eq:altern}
D^Y_N(\theta)=D^w_N(\theta)+D^Z_N(\theta)=\sqrt{n}D^w_d(\theta)+D^Z_N(\theta),
\end{equation}
where $D_N^
w(\theta)$, $D_d^w(\theta)$ and $D_N^Z(\theta)$ are the DFT's of $(Y_1,\ldots,Y_N)$, $(w_1,\ldots,w_d)$ and $(Z_1,\ldots,Z_N)$, respectively.

\begin{Proposition}\label{pr:consistff}
Assume model \eqref{e:model-G} and that $(Z_t)$ is $L^2$-$m$--approximable. Then if \ $\sum_{j=1}^q\|D^w_d(\vartheta_{j})\|^2>0$, we have that $T^{{\tt FTR_2}}\to\infty$ with probability 1. Moreover, if\ $\|D^w_d(\vartheta_1)\|^2> 0$, we have that $T^{{\tt FTR_1}}\to\infty$ with probability 1 ($N\to\infty$).
\end{Proposition}

Observe that
\[
\sum_{j=1}^q\|D^w_d(\vartheta_{j})\|^2
=\frac{1}{2}\sum_{t=1}^d\|w_t\|^2
=:\frac{d}{2}\mathrm{MSS}_\mathrm{sig}.
\]
Explicit forms for $\|D^w_d(\vartheta_1)\|^2$ and
$\sum_{j=1}^q\|D^w_d(\vartheta_j)\|^2$ when specialized to the
alternatives considered in models \eqref{e:model-Z}, \eqref{e:model-s}
and \eqref{e:model-G} are summarized in Table~\ref{tab:1}.  We infer
that if $(Z_t)$ satisfies Assumption~\ref{a:approx} with $r=2$,
then the functional tests based on
$T^{\tt FTR_2}$ (or
equivalently on $T^{\tt FAV}$) are consistent under the
alternatives in models \eqref{e:model-Z}, \eqref{e:model-s}
and \eqref{e:model-G}. The test based
on $T^{\tt FTR_1}$ is consistent for
model \eqref{e:model-Z}. It remains consistent for model
\eqref{e:model-s} provided $\alpha_1^2+\beta_1^2>0$,  and it is
consistent for model \eqref{e:model-G} if $\|D^w_d(\vartheta_1)\|^2>0$.

Consistency for the multivariate tests can be stated similarly.
Consider the representation \eqref{e:vector} of the projections.

\begin{Proposition} \label{pr:allconsist}
Consider model \eqref{e:model-G}  such that the noise is $L^2$-$m$--approximable. Let $\boldsymbol{D}_d^{\boldsymbol{w}}(\theta)=\frac{1}{\sqrt{d}}\sum_{t=1}^d  \boldsymbol{w}_{t}e^{-\mathrm{i} t\theta}$. If \ $\sum_{j=1}^q \|\boldsymbol{D}^{\boldsymbol{w}}_d(\vartheta_{j})\|^2>0$, we have that  $T^{\tt MEV_2}\to\infty$ and $T^{\tt MTR_2}\to\infty$ with probability 1. If \ $\|\boldsymbol{D}^{\boldsymbol{w}}_d(\vartheta_1)\|^2> 0$, we have that $T^{\tt MEV_1}\to\infty$ and $T^{\tt MTR_1}\to\infty$ with probability 1 ($N\to\infty$).
\end{Proposition}

As before, we can specialize the result to models \eqref{e:model-Z}
and \eqref{e:model-s}. Similar conditions as  for the functional case
are needed.

In Appendix~\ref{s:local}, we will present some results on local
consistency, i.e.\ we consider the case where the periodic signal
shrinks to zero with growing sample size. This study gives some
insight to the question in which situations a particular test can be
recommended. In this context we also refer to a numerical study in
Section~\ref{ss:localpower}.

\begin{table}
\begin{center}
\renewcommand{\arraystretch}{1.4}
\begin{tabular}{l|cc}
Model    & $\|D^w_d(\vartheta_1)\|^2$ & $\sum_{j=1}^q\|D^w_d(\vartheta_j)\|^2$ \\
\hline
\eqref{e:model-Z}     & $\frac{d}{4}\rho^2$    & $\frac{d}{4}\rho^2$     \\
\eqref{e:model-s}       &  $\frac{d}{4}(\alpha_1^2+\beta_1^2)$    & $\frac{d}{4}\sum_{k=1}^d(\alpha_k^2+\beta_k^2)$      \\
\eqref{e:model-G}      & $\|D^w_d(\vartheta_1)\|^2$    & $\frac{d}{2}\mathrm{MSS}_\mathrm{sig}$      \\
\end{tabular}
\caption{Explicit forms for $\|D^w_d(\vartheta_1)\|^2$ and $\sum_{j=1}^q\|D^w_d(\vartheta_j)\|^2$ when specialized to the alternatives
 in models \eqref{e:model-Z}, \eqref{e:model-s} and \eqref{e:model-G}.
\label{tab:1} }
\end{center}
\end{table}

\section{Application to pollution data} \label{s:realdata} We analyze
measurements of {\tt PM10} (particulate matter) and {\tt NO} (nitrogen
monoxide) in Graz, Austria, collected during one cold season, between
October 1, 2015 and March 15, 2016. Due to the geographic location of
Graz in a basin and unfavorable meteorological conditions (like
temperature inversion), the EU air quality standards are often not met
during the winter months. The measurement station is in the city
center (Graz-Mitte). Observations are available in the 30 minutes
resolution. The data were preprocessed in order to account for a few
missing values.  The measuring unit for both pollutants is $\mu
g/m^3$.  To improve the stability of our $L^2$ based methodology, we
follow \citetext{stadlober:hoermann:pfeiler:2008} and base our
investigations on the square-root transformed data. The resulting
discrete sample has been transformed into functional data objects with
the {\tt fda} package in {\tt R} using nine B-spline basis functions
of order four.

Our preliminary analysis, referred to
in  the Introduction,  was based on standard
ANOVA for daily averages, not taking into account the dependence of
the data. Viewing them as projections onto $v(u)\equiv 1$, we can
apply our tests $T^{\tt MEV_1}$ and $T^{\tt MEV_2}$ (or equivalently
$T^{\tt MTR_1}$ and $T^{\tt MTR_2}$ since $p=1$) adjusted for
dependence as explained in Section~\ref{s:dep}. The spectral density
of the daily averages is obtained as in Section~\ref{s:prac} with
$\gamma$ equal to the Bartlett kernel and $b_N=5$. The corresponding
$p$-values are given in Tables~\ref{tab:pm10} and \ref{tab:no} in the
rows tagged $v(u)\equiv 1$.

\begin{table}
\begin{center}
\begin{tabular}{l|cc|cc|cc}
   & $T^{\tt MEV_1}$ & $T^{\tt MTR_1}$ & $T^{\tt MEV_2}$ & $T^{\tt MTR_2}$  & $T^{\tt FTR_1}$ & $T^{\tt FTR_2}$\\
\hline
 FF\,\,\,\,\, ($100\%$) & & & &  & 0.180 & 0.104\\
 \hline
 $v(u)\equiv 1$ & 0.611 & 0.611 &  0.525 &0.525 & &\\
$p=1$ (71$\%$) &$ 0.564$  & $0.564$ & $0.492$ & $0.492$ & &\\
$p=2$ (82$\%$) &  $0.072$  & $0.083$ & $0.030$ & $0.031$ & &\\
$p=3$ (88$\%$) & $0.103$ & $0.091$& $0.071 $& $0.038 $& &\\
$p=5$ (96$\%$) & $< 10^{-4} $ & $<10^{-4} $ &$<10^{-5} $& $<10^{-4}$  & &
\end{tabular}
\end{center}
\caption{The $p$-values for {\tt PM10} data.
In parentheses,  the percentage of variance explained by the first  $p$
principal components on which the curves are projected.}
\label{tab:pm10}
\end{table}

\begin{table}
\begin{center}
\begin{tabular}{l|cc|cc|cc}
   & $T^{\tt MEV_1}$ & $T^{\tt MTR_1}$ & $T^{\tt MEV_2}$ & $T^{\tt MTR_2}$  & $T^{\tt FTR_1}$ & $T^{\tt FTR_2}$\\
\hline
 FF\,\,\,\,\, ($100\%$) & & & &  & 0.032 & 0.006\\
 \hline
 $v(u)\equiv 1$ & 0.305  & 0.305 & 0.099 & 0.099 & &\\
$p=1$ (68$\%$) &$ 0.247$  & $0.247$ & $0.076$& $0.076$ \\
$p=2$ (81$\%$) &  $0.495$ & $0.496$ &$0.204$ & $0.172$ \\
$p=3$ (87$\%$) & $<10^{-5}$& $< 10^{-5} $ & $<10^{-5}$& $< 10^{-5} $\\
$p=5$ (97$\%$) & $<10^{-5}$& $< 10^{-5} $ & $<10^{-5}$& $< 10^{-5} $
\end{tabular}
\end{center}
\caption{The $p$-values for {\tt NO} data.
In parentheses,  the percentage of variance explained by the first  $p$
principal components on which the curves are projected.}
\label{tab:no}
\end{table}

\begin{table}
\begin{center}
\begin{tabular}{l|cc|cc|cc}
   & $T^{\tt MEV_1}$ & $T^{\tt MTR_1}$ & $T^{\tt MEV_2}$ & $T^{\tt MTR_2}$  & $T^{\tt FTR_1}$ & $T^{\tt FTR_2}$\\
\hline
 FF\,\,\,\,\, ($100\%$) & & & &  & 0.556 & 0.737\\
 \hline
$p=1$ (67$\%$) &$ 0.582$ &  $0.582$ & $0.811$ & $0.811$ & & \\
$p=2$ (82$\%$) &  $0.171$  &$0.134$ & $0.335$& $0.356$ & & \\
$p=3$ (88$\%$) & $0.286$&  $0.117$& $0.493 $ &$0.307 $ & &\\
$p=5$ (96$\%$) & $0.515$&  $0.342$& $0.653$ & $0.654 $ & &
\end{tabular}
\end{center}
\caption{The $p$-values for residuals
in the regression of the {\tt NO} curves
onto the {\tt PM10} curves.}
\label{tab:res}
\end{table}

Next we conduct a periodicity analysis using the new tests. We compare
the fully functional (FF) tests $T^{\tt FTR_1}$ and $T^{\tt FTR_2}$
and the multivariate tests $T^{\tt MEV_1}$, $T^{\tt MEV_2}$, $T^{\tt MTR_1}$ and $T^{\tt MTR_2}$. Again,
we adjust the procedures for dependence, as explained in
Section~\ref{s:dep}. The spectral density and the covariance operator
(the latter is needed to compute principal components) are estimated
as described in Section~\ref{s:prac}. For the multivariate tests
we project data on the first $p$ principal components. We choose
$p=1,2,3$ and $p=5$---this choice guarantees that at least
$95\%$ of variance
are  explained for both data sets. The results are presented in
Tables~\ref{tab:pm10} and \ref{tab:no}.  It can be seen that the fully
functional procedures do give strong evidence of a weekly pattern for
{\tt NO}. From the multivariate tests we see that the first two
principal components do not pick up this periodic signal, but we get
strong evidence that it is concentrated in the third principal
component which is explaining about $6\%$ of the total variance.

For {\tt PM10} the situation is less clear cut. Though the $p$-values
are much smaller than in case of daily averages, the functional tests
are not significant at $5\%$ level. Looking at the multivariate tests
we do find a significant periodic signal if we project on higher order
principal components. These components explain only a relatively small
proportion of the total variance and hence the periodic pattern is not
easily made out on the global scale of the curves. The example
confirms that the projection based approach is more powerful in such
situations, with the drawback of being sensitive to the number  of
basis functions.

We conclude this illustrating example by regressing the {\tt NO}
curves onto the {\tt PM10} curves.  The function on function
regression is done using the $B$-spline expansion, see e.g.
\citetext{ramsay:hooker:graves:2009}.  We analyze the residual curves.
The $p$-values are summarized in Table~\ref{tab:res}. None of our
tests yields significant evidence that there remains a weekly
periodicity in the residual curves.  This indicates that in
Graz--Mitte, the sources for both pollutants {\tt PM10} and {\tt NO}
are the same.

\section{Assessment based on simulated data} \label{s:sim} Our goal is
to assess empirical rejection rates of our  tests,  under
$\mathcal{H}_0$ as well as under $\mathcal{H}_A$,  in some realistic
finite sample settings.  For this purpose, we consider the functional
time series of {\tt PM10}, pre-processed as explained in
Section~\ref{s:realdata}. We remove the weekday mean curves $\widehat
w_k$, $1\leq k\leq 7$, (from every Monday curve, we remove Monday's
mean $\widehat w_1$, etc.).  We then generate  series of functional
data by bootstrapping (with replacement) the times series of these
residuals. The resulting i.i.d.\ data are denoted $\epsilon_{t}, \;
t=1,\ldots,M$.  Next we generate dependent errors by setting
\[
Z_{t} = \epsilon_{t} + a_{1} \epsilon_{t-i} + \ldots
+ a_{5} \epsilon_{t-5}, \quad t=6,\ldots,M,
\]
where $a_{k} = e^{-k}$ are scalar coefficients. We chose $M=215$ and
$425$ so that the length of the time series, $N$, is $210$ and $420$.
Then we run our tests with the procedures adjusted for dependence as explained in
Section~\ref{s:dep}. Our estimator of the spectral density is defined
by (\ref{e:win2}) with $b_{N} = \lfloor N^{1/3} \rfloor$. The results
are shown in Table~\ref{tab:empdep}. We see that the fully functional
tests have a very good empirical size.
Also the multivariate tests, where we projected on the first $p$
eigenfunctions of the data, perform well, especially for smaller
values of $p$.
We have experimented with other simulation setups,  not reported
here. Throughout,  we found that the fully
functional tests are more reliable than  the multivariate tests in terms of their
empirical size. This is most likely  explained by the fact that the fully
functional methods are not very sensitive to the effect of estimation
errors for small eigenvalues. The distributions of
$\mathrm{HExp}(\lambda_1,\lambda_2,\ldots)$ and
$\mathrm{HExp}(\hat\lambda_1,\hat\lambda_2,\ldots)$ are typically
close, because they mainly depend on a few large eigenvalues for which
the relative estimation error is small.  For the multivariate tests,
eigenvalues enter as reciprocals. If $\lambda_k$ is close to
$\hat\lambda_k$, it does not  necessarily mean that $1/\lambda_k$ and
$1/\hat\lambda_k$ are close,  if the eigenvalues are small.

\begin{table}
\begin{center}
\resizebox{\columnwidth}{!}{
\begin{tabular}{l|cccc|cc||cccc|cc}
 & \multicolumn{6}{c||}{$\alpha=5\%$} &   \multicolumn{6}{c}{$\alpha=10\%$} \\
   & $T^{\tt MEV_1}$ & $T^{\tt MTR_1}$& $T^{\tt MEV_2}$ & $T^{\tt MTR_2}$  & $T^{\tt FTR_1}$ & $T^{\tt FTR_2}$& $T^{\tt MEV_1}$ & $T^{\tt MTR_1}$  & $T^{\tt MEV_2}$ & $T^{\tt MTR_2}$ & $T^{\tt FTR_1}$ & $T^{\tt FTR_2}$\\
 \hline
FF &  &  &  &  & 5.1 & 5.0 &  &  &  &  & 9.2 & 8.3 \\
   &  &  &  &  &  5.9 & 4.7 &  &  &  &  & 10.8 & 9.9 \\ \hline
  $p=1$ & 5.0 & 5.0 & 3.9 & 3.9 &  &  & 9.6 & 9.6 & 8.7 & 8.7 &  &  \\
             & 5.9 & 5.9 & 4.4 & 4.4 &  &  & 9.8 & 9.8 & 9.9 & 9.9 &  &  \\ \hline
  $p=2$ & 6.4 & 6.5 & 4.3 & 3.9 &  &  & 11.2 & 11.4 & 8.7 & 8.0 &  &  \\
           & 6.0 & 6.0 & 5.4 & 4.1 &  &  & 10.6 & 10.6 & 9.8 & 9.1 &  &  \\ \hline
  $p=3$ & 6.8 & 5.9 & 3.8 & 3.9 &  &  & 12.2 & 11.8 & 7.9 & 6.9 &  &  \\
   & 5.8 & 5.7 & 5.5 & 4.2 &  &  & 10.6 & 11.2 & 9.6 & 7.9 &  &  \\ \hline
  $p=5$ & 7.4 & 8.7 & 6.4 & 6.2 &  &  & 15.5 & 15.9 & 11.6 & 11.4 &  &  \\
           & 6.7 & 7.9 & 6.4 & 5.7 &  &  & 12.5 & 12.9 & 12.2 & 10.7 &  &
\end{tabular} }
\end{center}
\caption{Empirical size (in $\%$) at the
nominal level $\alpha$ of $5\%$ and $10\%$ for dependent
 time series with sample sizes $N=210$ (top rows) and $N=420$
(bottom rows). Results are based on 1000 Monte Carlo simulation runs.
 \label{tab:empdep}}
\end{table}

To see how well the tests can detect a realistic alternative, we use
the same data generating process as above and periodically add the
weekday means $\widehat w_1,\ldots\widehat w_7$ to the stationary
noise, say $(Z_t)$.  We thus get the series $V_t=\widehat w_{(t)}+Z_t$
where $(t)= t\,\text{mod}\, 7$ with the convention that $\widehat
w_{0}=\widehat w_{7}$. This construction entails that we are in the
setting of Model~\eqref{e:model-G} and hence, in view of
Theorem~\ref{th:LR}, we expect the multi-frequency and trace based
tests to be most powerful.  This is confirmed in
Table~\ref{tab:empdepalt} where we show empirical rejection rates.
The power of the eigenvalue based tests is very similar.  We see again that, in terms of power, the multivariate
tests perform best, once we project onto an appropriate
  subspace.  Let us note that in this example
$\text{MSS}_{\mathrm{sig}}=\frac{1}{7}\sum_{k=1}^7\|\widehat
w_k-\overline{\widehat{w}}\|^2\approx 0.1$ and $E\|Z_k\|^2\approx
3.1$. Given the relatively small signal-to-noise ratio, we can
see that overall the tests perform very well in finite samples.

\begin{table}
\begin{center}
\resizebox{\columnwidth}{!}{
\begin{tabular}{l|cccc|cc||cccc|cc}
 & \multicolumn{6}{c||}{$N=210\%$} &   \multicolumn{6}{c}{$N=420\%$} \\
   & $T^{\tt MEV_1}$ & $T^{\tt MTR_1}$& $T^{\tt MEV_2}$ & $T^{\tt MTR_2}$  & $T^{\tt FTR_1}$ & $T^{\tt FTR_2}$& $T^{\tt MEV_1}$ & $T^{\tt MTR_1}$  & $T^{\tt MEV_2}$ & $T^{\tt MTR_2}$ & $T^{\tt FTR_1}$ & $T^{\tt FTR_2}$\\
 \hline
FF &  &  &  &  & 39.2 & 72.9 &  &  &  &  & 82.6 & 99.9 \\  \hline
  $p=1$ & 14.3 & 14.3 & 26.5 & 26.5 &  &  & 21.7 & 21.7 & 56.6 & 56.6 &  &  \\  \hline
  $p=2$ & 50.2 & 50.3 & 89.4 & 89.9 &  &  & 76.9 & 77.8 & 99.7 & 99.8 &  &  \\  \hline
  $p=3$ & 73.4 & 76.4 & 96.4 & 98.1 &  &  & 92.2 & 94.7 & 100 & 100 &  &  \\  \hline
  $p=5$ & 99.22 & 99.6 & 100 & 100 &  &  & 100 & 100 & 100 & 100 &  &  \\
\end{tabular} }
\end{center}
\caption{Empirical rates (in $\%$) when testing at
nominal level $\alpha$ of $5\%$. Results are based on 1,000 Monte Carlo simulation runs.
\label{tab:empdepalt}}
\end{table}

The rejection rates reported in this section are based on a specific example and a
specific estimator of the covariance structure, the same one as used in
Section~\ref{s:realdata}. To gain insights into the {\em asymptotic}
 rejection rates, we perform in Section~\ref{ss:localpower}
a numerical study which does not use a specific estimator,
but assumes a known covariance structure. This approach allows us
to isolate the effect of estimation from the intrinsic properties of
the tests.

\section{Local asymptotic  power}\label{ss:localpower}
A power study must necessarily involve a larger number of data
generating processes (DGP's) which satisfy the  various alternatives
considered in this paper. We consider here
18 DGP's,  indexed by the period $d=7, 31$ and
$i, j = 1, 2, 3$, which have the
general form
\begin{equation} \label{e:alt}
Y_t(u)= s_t^{(i,d)} \left(\sum_{k=1}^9\psi_k^{(j)} v_k(u)\right)
+ \sum_{k=1}^9 z_{t,k}v_k(u), \quad i,j=1,2,3.
\end{equation}
The $v_1, v_2, \ldots, v_9$ are orthonormal basis functions.
We note right away that the results do not depend on
what specific form the $v_k$ take, as long as they are orthonormal.
The  $s_t^{(i,d)}$ is a real $d$-periodic signal with $\sum_{t=1}^d
s_t^{(i,d)}=0$ and $\psi_k^{(j)}$ are real coefficients. The exact
specifications are given below. The variables
$\boldsymbol{z_t}=(z_{t,1},z_{t,2},\ldots,z_{t,9} )^\prime$ are
i.i.d.\ Gaussian vectors with zero mean and covariance
$\mathrm{diag}(1,2^{-1},2^{-2},2^{-3},\ldots,2^{-8})$. Then $(Y_t)$
follows the functional model \eqref{e:model-s} with
$w(u)=w^{(j)}(u)=\sum_{k=1}^9\psi_k^{(j)} v_k(u)$. Our assumptions
imply that the $v_k$ are the functional
principal components of $Y_t$.  We
consider periods of length $d=7$ and $d=31$. For the periodic signal
we consider the following variants
\begin{align*}
&s_t^{(1,d)}= \cos(2\pi t/d);\\
&s_t^{(2,d)}= I\left\{1\leq t\leq 2(d-1)/3\right\}-2I\left\{(2d-1)/3+1<t\leq d\right\}\quad \text{for} \quad 1\leq t\leq d;\\
&s_t^{(3,d)}=v_t-\bar{v},\quad\text{where}\quad v_t\stackrel{\mathrm{i.i.d.}}{\sim} N(0,1)\quad \text{for} \quad 1\leq t\leq d.
\end{align*}
We consider the following parameters $\boldsymbol{\psi}^{(j)}=(\psi_1^{(j)},\ldots,\psi_9^{(j)})^\prime$:
\begin{align*}
&\boldsymbol{\psi}^{(1)}=(1,0,0,0,0,\ldots,0)^\prime;\\
&\boldsymbol{\psi}^{(2)}\propto (1,2^{-1/2},2^{-1},2^{-3/2},\ldots,2^{-4})^\prime;\\
&\boldsymbol{\psi}^{(3)}=(0,0,0,1,0,\ldots,0)^\prime.
\end{align*}
The vectors $\boldsymbol{\psi}^{(j)}$ determine $w(u)$ and are scaled
to unit length. Under parametrization $\boldsymbol{\psi}^{(1)}$
($\boldsymbol{\psi}^{(3)}$),  we have $w(u)$ varying in direction of the
first (fourth) principal component, while under
$\boldsymbol{\psi}^{(2)}$ we take into account all principal
components. The DGP is determined by the pair
$(\boldsymbol{\psi}, s) =(\boldsymbol{\psi}, s^{(i,d)})$.

 We  study  the  \emph{local asymptotic power} functions defined by
\[
LP(x|\boldsymbol{\psi},s)
=\lim_{N\to\infty} P\big(T_N>q_{0.95}\,|\,\text{DGP is $(\boldsymbol{\psi},\frac{x}{\sqrt{N}}\,s)$}\big),
\]
where $T_N$ stands for one of the test statistics we derived, and
$q_{0.95}$ is its (asymptotic) 95th quantile under the null. We use a
superscript to indicate which statistic is used, for example, $LP^{\tt
  MEV_2}$, $LP^{\tt FTR_1}$, etc. It can be easily seen that if the
covariance operator $\Gamma$ is known, then, due to our Gaussian
setting, $P\big(T_N>q_{0.95}\,|\,\text{DGP is
  $(\boldsymbol{\psi},\frac{x}{\sqrt{N}}\,s)$}\big)$ does not  dependent
on $N$ for any of our tests. Since we let $N\to\infty$,  we can use a
Slutzky argument and compute $LP(x|\boldsymbol{\psi},s)$ directly with
the true $\Gamma$. It is not obvious how to obtain closed analytic forms for
$LP(x|\boldsymbol{\psi},s)$ and hence we compute them numerically by
Monte-Carlo simulation based on 5,000 replications.

\begin{enumerate}
\item \emph{Comparing $T^{\tt MEV_2}$ and $T^{\tt MTR_2}$: eigenvalue v.s.\ trace based test statistic.} \medskip

  We project data onto the space spanned by $v_1,\ldots,v_5$, which
  guarantees that at least $95\%$ of variance are explained.  In
  Figure~\ref{fig:powerq2} the asymptotic local power curves $LP^{\tt
    MEV_2}(x|\boldsymbol{\psi}^{(2)},s^{(2,d)})$ and $LP^{\tt
    MTR_2}(x|\boldsymbol{\psi}^{(2)},s^{(2,d)})$ with $d=7$ and $d=31$
  are presented. We have done the same exercise with
  $\boldsymbol{\psi}^{(1)}$ and $\boldsymbol{\psi}^{(3)}$ and obtained
  very similar results.

\item \emph{Comparing $T^{\tt MEV_1}$ and $T^{\tt MEV_2}$: test for sinusoidal v.s.\ test for general periodic pattern.} \medskip

  We project again onto $v_1,\ldots,v_5$.  In
  Figure~\ref{t:simuEV1EV2} the asymptotic local power curves $LP^{\tt
    MEV_1}(x|\boldsymbol{\psi}^{(2)},s)$ and $LP^{\tt
    MEV_2}(x|\boldsymbol{\psi}^{(2)},s)$ are shown with $s=s^{(2,7)}$
  (left panel), $s=s^{(2,31)}$ (middle panel) and $s=s^{(3,7)}$ (right
  panel). We see that the LR-test for the simpler model
  \eqref{e:model-Z} can significantly outperform the LR-test for model
  \eqref{e:model-s} even if $s_t^{(2,d)}$ is not sinusoidal. However,
  the conclusion is very different if $s$ is more erratic. When
  $s=s_t^{(3,7)}$, then $T^{\tt MEV_2}$ becomes a lot more powerful
  than $T^{\tt MEV_1}$. Simulations not reported here show that the
  above described effects become stronger the larger we choose the
  period $d$. This finding is supported by
  Proposition~\ref{pr:localpower} in our supplement, which provides a
  theoretical result on local consistency.
\item \emph{Comparing $T^{\tt MEV_1}$ and $T^{\tt FTR_1}$: projection based v.s.\ fully functional method.}\medskip

Now the objective is to compare the projection based methods with the fully functional ones. By fixing $s=s^{(1,7)}$ we focus on the simple model \eqref{e:model-Z}. The local power curves $LP^{\tt FTR_1}(x|\boldsymbol{\psi}^{(i)},s^{(1,7)})$ and $LP^{\tt MEV_1}(x|\boldsymbol{\psi}^{(i)},s^{(1,7)})$ for values $p=1,2,3$ and $p=5$ and $i=1,2,3$ are shown in Figure~\ref{t:simuMEVFTR}. We see that the fully functional test performs well in all settings. Not surprisingly, the better the basis onto which we project describes $w(u)$, the better the projection based method becomes. For all DGP's $(\boldsymbol{\psi}^{(i)},s^{(1,7)})$, $i=1,2,3$, there is one projection based test that outperforms the functional one. The disadvantage of the projection method is, however, its sensitivity with respect to the chosen basis. For example, while for DGP $(\boldsymbol{\psi}^{(1)},s^{(1,7)})$ the test with $p=1$ is performing best, it is the least powerful for DGP's $(\boldsymbol{\psi}^{(2)},s^{(1,7)})$ and
$(\boldsymbol{\psi}^{(3)},s^{(1,7)})$.
\end{enumerate}

\begin{figure}[h]
\begin{tabular}{c}
\includegraphics[scale=0.33]{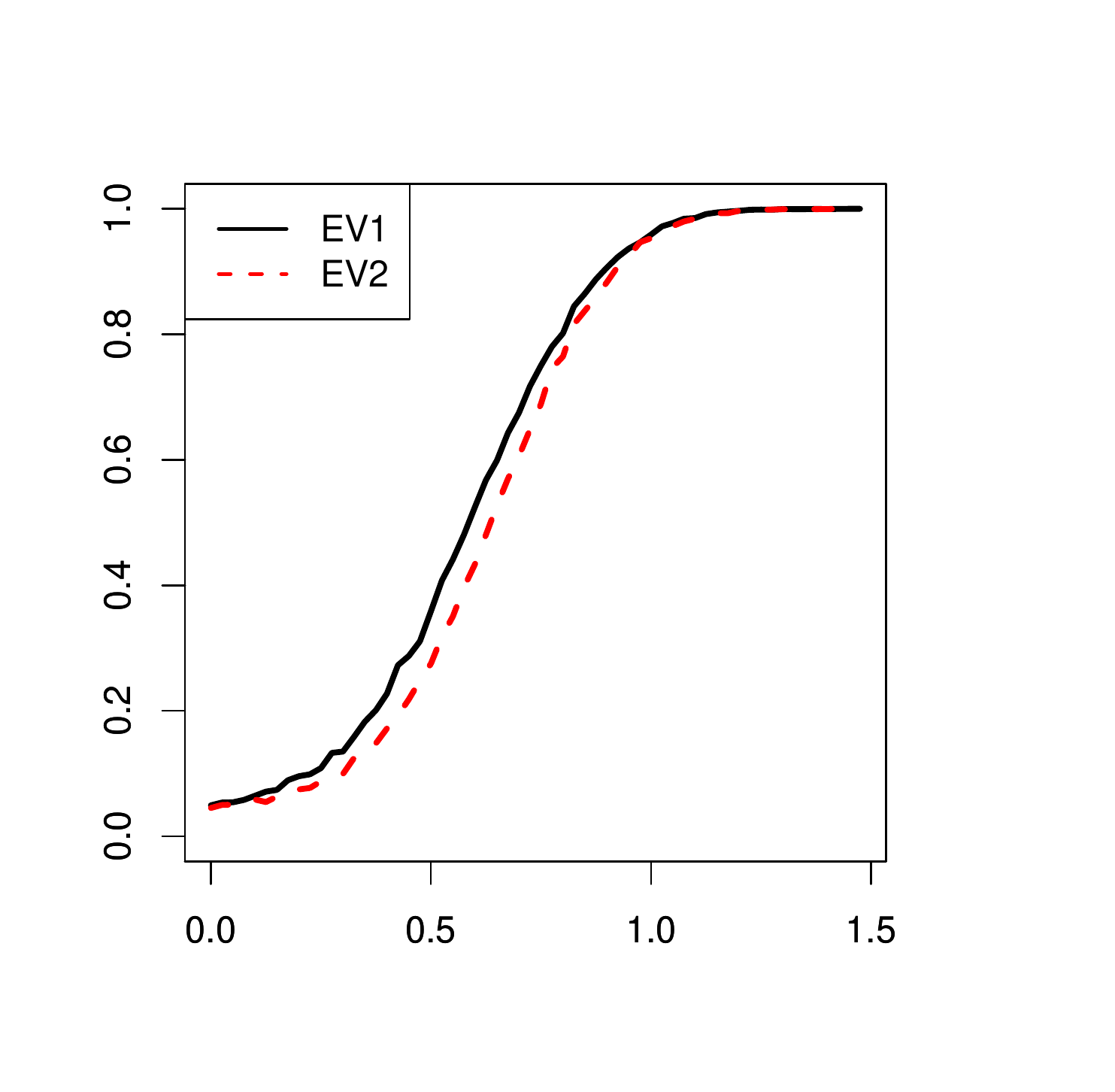}\hfill
 \includegraphics[scale=0.33]{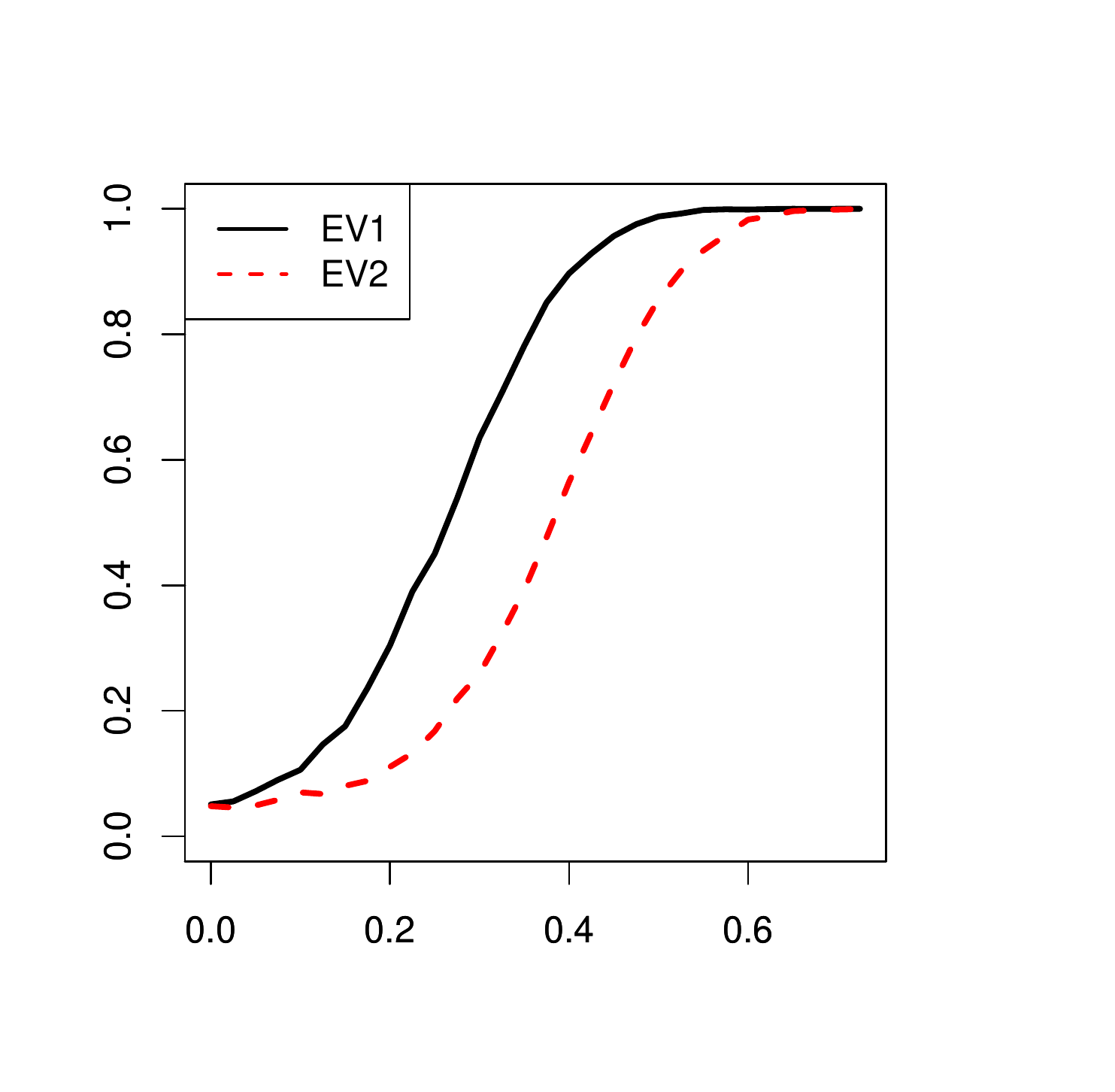}\hfill
 \includegraphics[scale=0.33]{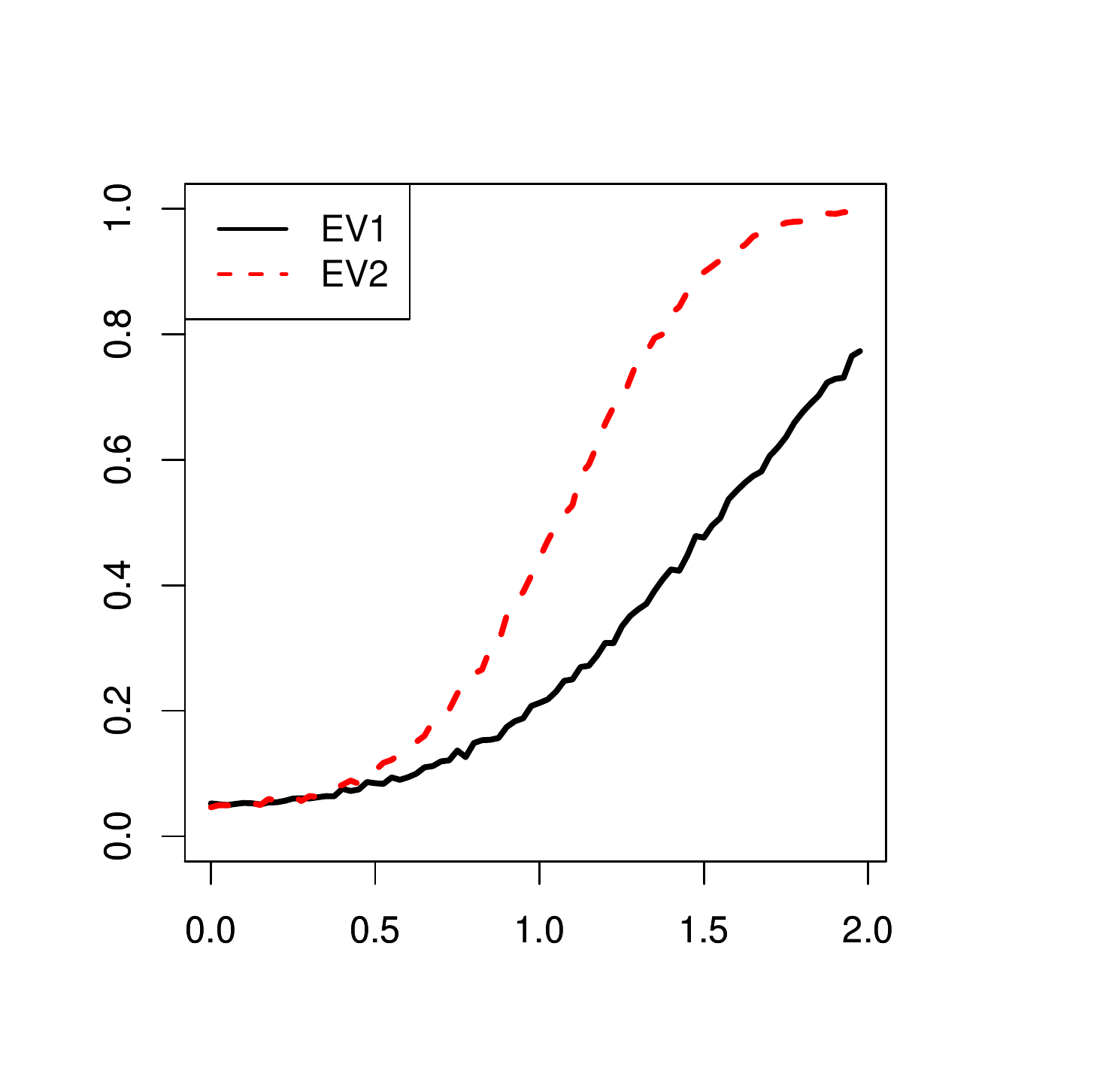}
\end{tabular}
\caption{Local power curves $LP^{\tt MEV_1}(x|\boldsymbol{\psi}^{(2)},s)$ (EV1) and $LP^{\tt MEV_2}(x|\boldsymbol{\psi}^{(2)},s)$ (EV2) with $s=s^{(2,7)}$ (left panel), $s=s^{(2,31)}$ (middle panel) and $s=s^{(3,7)}$ (right panel), with the realization of $(s^{(3,7)}_t\colon 1\leq t\leq 7)=(-0.24,  0.42, -1.69,  0.37,  0.07,  1.12, -0.05)$. }
\label{t:simuEV1EV2}
\end{figure}
\begin{figure}[h]
\begin{tabular}{c}
\includegraphics[scale=0.27]{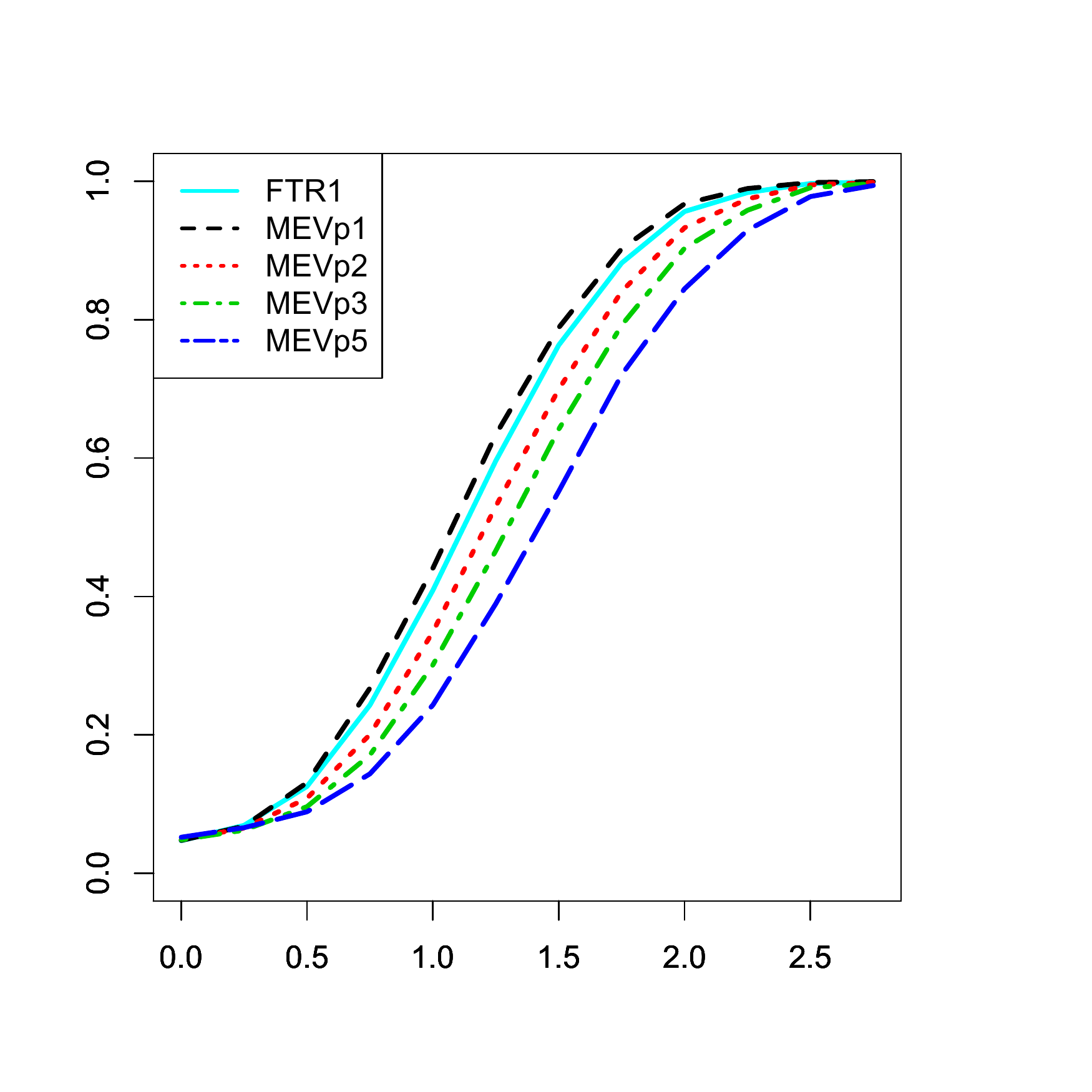}\hfill
 \includegraphics[scale=0.27]{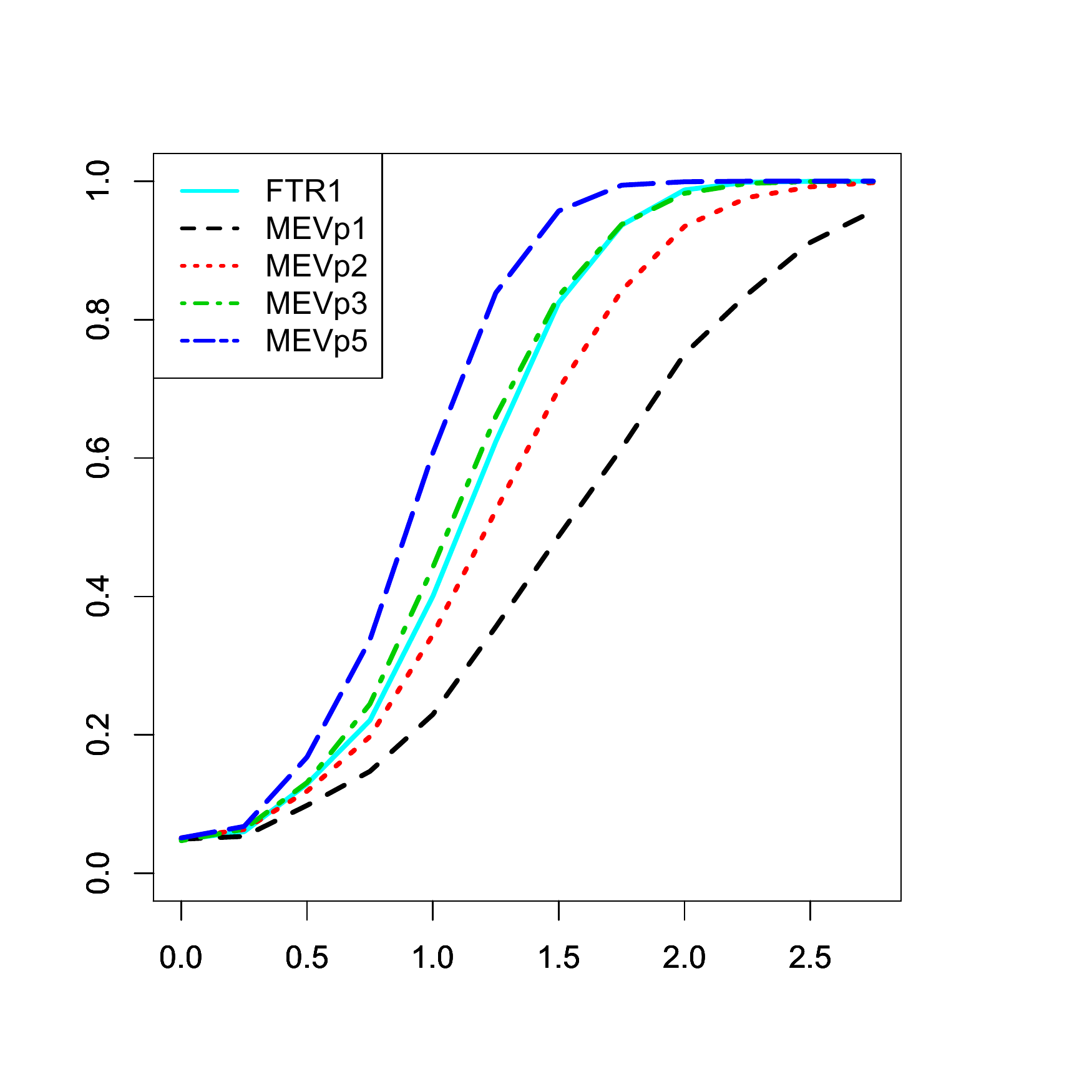}\hfill
 \includegraphics[scale=0.27]{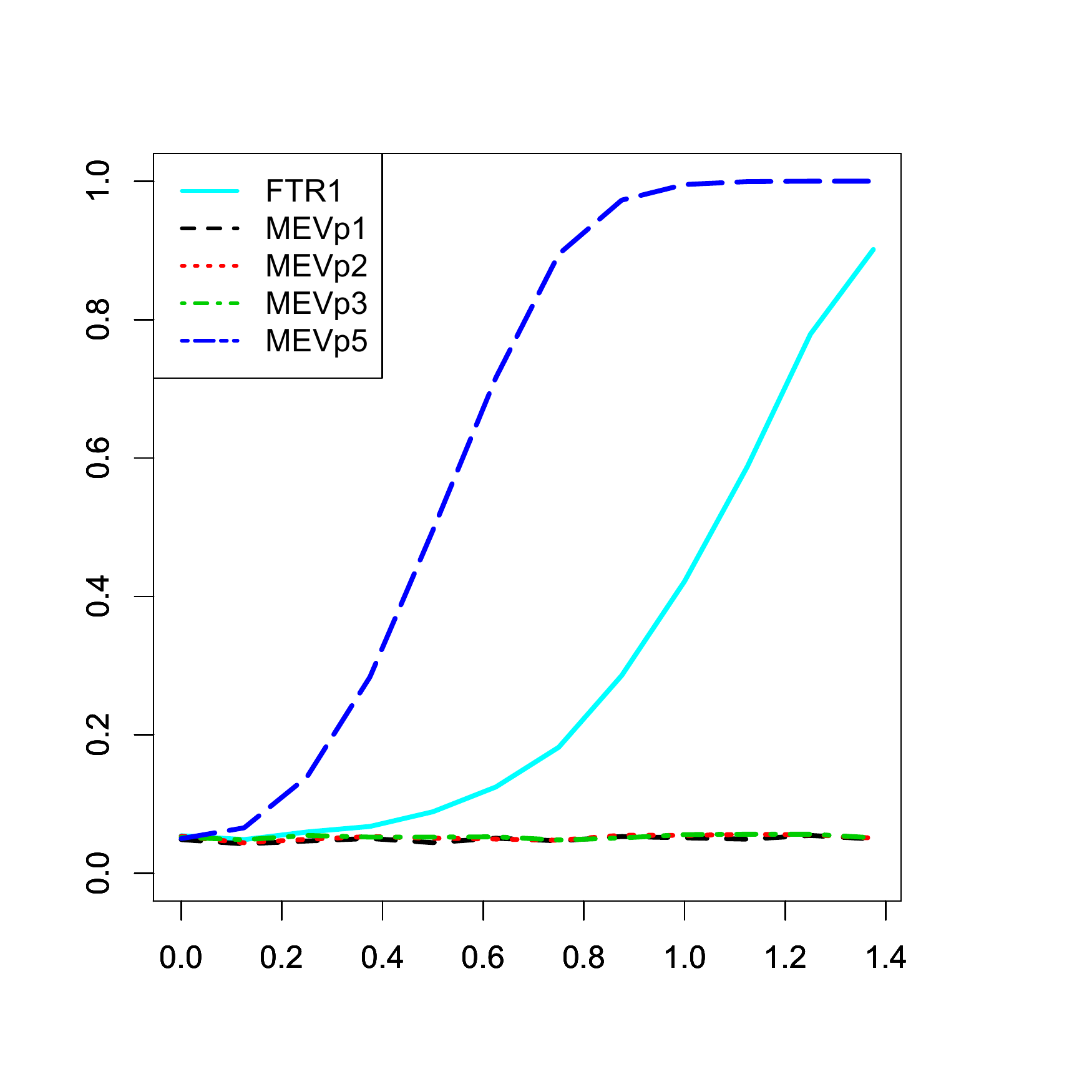}
\end{tabular}
\caption{Local power curves $LP^{\tt FTR_1}(x|\boldsymbol{\psi}^{(i)},s^{(1,7)})$ and $LP^{\tt MEV_1}(x|\boldsymbol{\psi}^{(i)},s^{(1,7)})$ for values $p=1,2,3$ and $p=5$ and $i=1$ (left panel), $i=2$ (middle panel) and $i=3$ (right panel). }
\label{t:simuMEVFTR}
\end{figure}

\clearpage
\section{Summary} \label{s:con} We have proposed several tests for
detecting periodicity in functional time series which fall into two
broad categories which we refer to as multivariate and fully
functional approaches.  Our tests are motivated by the Gaussian
likelihood ratio approach, and, in general, have the expected power
advantage for multivariate time series for which other tests exist.
Allowing general weak dependence of errors is also new even for
multivariate data. For functional data, all tests are new.  In what
follows we summarize the main conclusions of our work.

\begin{itemize}
\item Generally,  the functional approach is a more
adequate and safer option. The multivariate approach can be more
powerful, but it  is sensitive to the choice of the subspace on which the data
are projected.
\item If the signal is close to sinusoidal, then the simple single
  frequency test is more powerful, otherwise the opposite is true. The
  effect becomes stronger with length of the period. This empirical
  finding is theoretically confirmed in Appendix~\ref{s:local} of the
  supplement.
\item For the multivariate tests we have seen that the eigenvalue
  statistics can have a considerable power advantage over the
  traditionally used trace based statistics. Theoretically, we have
  shown that $T^{\tt MEV_1}$ and $T^{\tt MEV_2}$ can be justified by a
  LR procedure when the periodic signal is proportional to a single
  function $w(u)$. There exists an easy algorithm to compute critical
  values.
\end{itemize}
If no prior knowledge on the periodic component is available, we
recommend to use the ANOVA based approach or to base the decision on
more than one test. Simultaneous acceptance or simultaneous rejection
by several tests will lend confidence in the conclusion.

\bibliographystyle{oxford3}
\renewcommand{\baselinestretch}{0.9}
\small
\bibliography{per_06_13}

\renewcommand{\baselinestretch}{1.0}

\normalsize\newpage

\centerline{\Large Supplemental material}

\appendix
\section{Discussion of Assumption~\ref{ass:d}} \label{s:gen-d}
In this section, we explain how the test procedure should be adapted
when $d$ is even.  If $d=2$, then we can look at the lag-1 differenced series $\Delta Y_t=Y_{t}-Y_{t-1}$ and the problem boils down to testing for a zero mean of $(\Delta Y_t)$. So assume that $d\geq 4$. The tests $T^{{\tt MTR_1}}$  and $T^{{\tt FTR_1}}$  remain unchanged. The test statistics tests $T^{{\tt MTR_2}}$  and $T^{{\tt FTR_2}}$  have to be defined in a slightly different way.  We now set $r=(d-2)/2$ and replace $\boldsymbol{A}(\vartheta_1,\ldots,\vartheta_q)$ by
\begin{equation}
\boldsymbol{B}(\vartheta_1,\ldots,\vartheta_r)= \left[
\boldsymbol{R}(\vartheta_1),\ldots,
\boldsymbol{R}(\vartheta_r),
\boldsymbol{C}(\vartheta_1),\ldots,
\boldsymbol{C}(\vartheta_r),
\boldsymbol{R}(\theta_{N/2})
\right]^\prime.
\end{equation}
The corresponding tests in Theorem~\ref{th:LR}  have to be changed to
\begin{align*}
T^{{\tt MEV}_2}&:=\big\|\boldsymbol{B}(\vartheta_1,\ldots,\vartheta_r)\boldsymbol{\Sigma}^{-1}\boldsymbol{B}^\prime(\vartheta_1,\ldots,\vartheta_r)\big\| > q_{1-\ag}[\|\bW_p(d-1)\|/2],\\
T^{{\tt MTR_2}}&:=\big\|\boldsymbol{B}(\vartheta_1,\ldots,\vartheta_r)\boldsymbol{\Sigma}^{-1}\boldsymbol{B}^\prime(\vartheta_1,\ldots,\vartheta_r)\big\|_\mathrm{tr}>q_{1-\ag}[ \chi^2_{p(d-1)}/2],
\end{align*}
respectively.
The fully functional test \eqref{e:tr-fq} becomes
\[
T^{{\tt FTR_2}}:=\sum_{k=1}^r\big(\|R(\vartheta_k)\|^2+ \|C(\vartheta_k)\|^2\big)+\|R(\theta_{N/2})\|^2>q_{1-\ag}
\left [ \sum_{k=1}^{r}\Xi_k+\Theta \right ],
\]
where $\Xi_k\stackrel{\mathrm{i.i.d.}}{\sim} \mathrm{HExp}(\lambda_1,\lambda_2,\ldots)$ and $\Theta\sim \frac{1}{2}\sum_{\ell \geq 1}\lambda_\ell \chi^2(\ell)$ with i.i.d.\ $\chi^2$ variables $\chi^2(\ell)$. The $\lambda_\ell$ are the eigenvalues of $\mathrm{Var}(Y_1)$ and $\Theta$ is independent of $\Xi_1,\ldots,\Xi_r$. Accordingly, the analogue type changes are required in the setting of Section~\ref{s:dep}.

\section{Details of the derivation of the tests of Section~\ref{s:tests}}
\label{s:ta}

\begin{proof}[Proof of Theorem~\ref{th:LR}]
We only derive the LR test $T^{{\tt MEV_2}}$. Derivation of $T^{{\tt MEV_1}}$ is similar and $T^{{\tt MTR_2}}$ can be deduced from Proposition~\ref{pr:anova_tr}.

Set $\boldsymbol{A}=\boldsymbol{A}(\vartheta_1,\ldots,\vartheta_q)$ and $\boldsymbol{\gamma}=(\alpha_1,\ldots,\alpha_q,\beta_1,\ldots,\beta_q)^\prime$ and let $\boldsymbol{e}_1$ be the eigenvector associated to the largest eigenvalue of $\boldsymbol{\Sigma}^{-1/2}\boldsymbol{A}^\prime\boldsymbol{A}\boldsymbol{\Sigma}^{-1/2}$. The proof will reveal that  the MLEs of $\boldsymbol{w}$ and $\boldsymbol{\gamma}$ are given as
$\widehat{\boldsymbol{w}}=\boldsymbol{\Sigma}^{1/2} \boldsymbol{e}_1$ and $\widehat{\boldsymbol{\gamma}}=\frac{2}{\sqrt{N}}\boldsymbol{A} \boldsymbol{\Sigma}^{-1/2} \boldsymbol{e}_1$, respectively.
We notice that $\widehat{\boldsymbol{w}}$ and $\widehat{\boldsymbol{\gamma}}$ are not unique.  For any $x>0$, the pair $(\widehat{\boldsymbol{w}},\widehat{\boldsymbol{\gamma}})$ maybe replaced by $(x\widehat{\boldsymbol{w}},\widehat{\boldsymbol{\gamma}}/x)$. If the largest eigenvalue of $\boldsymbol{\Sigma}^{-1/2}\boldsymbol{A}^\prime\boldsymbol{A}\boldsymbol{\Sigma}^{-1/2}$ has multiplicity one, then uniqueness can be obtained imposing $\|\widehat{\boldsymbol{w}}\|=1$.

Define $\boldsymbol{c}_i=\big(\cos(t\vartheta_i)\colon 1\leq t\leq N\big)^\prime$ and $\boldsymbol{s}_i=\big(\sin(t\vartheta_i)\colon 1\leq t\leq N\big)^\prime$. Moreover we set
$\boldsymbol{w}^\mathrm{sc} = \bSig^{-1/2}\boldsymbol{w}$ and
$\boldsymbol{Y}_{t}^{\mathrm{sc}}  = \bSig^{-1/2}\boldsymbol{Y}_{t}$. First we observe that under $\mathcal{H}_0$ as well as under $\mathcal{H}_A$ the MLE for $\boldsymbol{\mu}=\overline{\boldsymbol{Y}}$.
Using representation \eqref{e:rep-s} of the underlying model, mutual orthogonality of the vectors $\boldsymbol{c}_1,\ldots,\boldsymbol{c}_q, \boldsymbol{s}_1,\ldots,\boldsymbol{s}_q$ and $\|\boldsymbol{c}_i\|=\|\boldsymbol{s}_i\|=N/2$, we obtain with simple algebra that the log-likelihood ratio  is given by \ $\max_{\boldsymbol{w}^\mathrm{sc},\boldsymbol{\gamma}}\ell_N(\boldsymbol{w}^\mathrm{sc},\boldsymbol{\gamma})$, where
$$
\ell_N(\boldsymbol{w}^\mathrm{sc},\boldsymbol{\gamma})=\big(\boldsymbol{w}^\mathrm{sc}\big)^\prime \boldsymbol{\Sigma}^{-1/2}\boldsymbol{A}^\prime\boldsymbol{\gamma}-\frac{\sqrt{N}}{4}\|\boldsymbol{w}^\mathrm{sc}\|^2\|\boldsymbol{\gamma}\|^2.
$$
Since for any $x>0$ we have
$\ell_N(\boldsymbol{w}^\mathrm{sc},\boldsymbol{\gamma})=\ell_N(x\boldsymbol{w}^\mathrm{sc},\boldsymbol{\gamma}/x)$, we can assume that $\|\boldsymbol{w}^\mathrm{sc}\|=1$.
Under this constraint we maximize
$$
\ell_N(\boldsymbol{w}^\mathrm{sc},\boldsymbol{\gamma})=\big(\boldsymbol{w}^\mathrm{sc}\big)^\prime \boldsymbol{\Sigma}^{-1/2}\boldsymbol{A}^\prime\boldsymbol{\gamma}-\frac{\sqrt{N}}{4}\|\boldsymbol{\gamma}\|^2.
$$
For given $\boldsymbol{w}^\mathrm{sc}$ we obtain the maximizer $\widehat{\boldsymbol{\gamma}}=\frac{2}{\sqrt{N}}\boldsymbol{A}\boldsymbol{\Sigma}^{-1/2} \boldsymbol{w}^\mathrm{sc}$ and
\begin{align*}
\max_{\boldsymbol{w}^\mathrm{sc}\colon \|\boldsymbol{w}^{\mathrm{sc}}\|=1 }\ell(\widehat{\boldsymbol{\gamma}},\boldsymbol{w}^\mathrm{sc})&=\frac{1}{\sqrt{N}}\max_{\boldsymbol{w}^\mathrm{sc}\colon \|\boldsymbol{w}^{\mathrm{sc}}\|=1 }\big(\boldsymbol{w}^\mathrm{sc}\big)^\prime \boldsymbol{\Sigma}^{-1/2}\boldsymbol{A}^\prime\boldsymbol{A}\boldsymbol{\Sigma}^{-1/2}\boldsymbol{w}^\mathrm{sc}\\
&=\frac{1}{\sqrt{N}}\|\boldsymbol{\Sigma}^{-1/2}\boldsymbol{A}^\prime\boldsymbol{A}\boldsymbol{\Sigma}^{-1/2}\|=\frac{1}{\sqrt{N}}\|\boldsymbol{A}\boldsymbol{\Sigma}^{-1}\boldsymbol{A}^\prime\|.
\end{align*}
Moreover, the maximizing $\boldsymbol{w}^\mathrm{sc}=\boldsymbol{e}_1$.
\end{proof}

\begin{Lemma}\label{le:orth}
Under Assumption~\ref{ass:gauss} and $\mathcal{H}_0$ we have that $$\big(R(\theta_1),\ldots,R(\theta_m),C(\theta_1)\cdots,C(\theta_m)\big)$$ are i.i.d.\ elements in $H$. We have $R(\theta_1)\sim \mathcal{N}_H(0,\frac{1}{2}\Gamma)$. The analogous result holds if the functions $R(\theta_i)$ and $C(\theta_i)$ are replaced by their projections $\mathbf{R}(\theta_i)$ and  $\mathbf{C}(\theta_i)$.  In this case we have $\mathbf{R}(\theta_i)\sim \mathcal{N}_p(0,\frac{1}{2}\boldsymbol{\Sigma})$.
\end{Lemma}
\begin{proof}
Clearly, the vectors are jointly Gaussian. Recall that $\theta_j\in\Theta_N$ are the fundamental frequencies. Set $\boldsymbol{c}(\theta)=\big(\cos(t\theta)\colon 1\leq t\leq N\big)^\prime$ and $\boldsymbol{s}(\theta)=\big(\sin(t\theta)\colon 1\leq t\leq N\big)^\prime$, then the result is a simple consequence of the fact that the vectors $(\boldsymbol{c}(\theta)\colon \theta\in\Theta_N)$ and $(\boldsymbol{s}(\theta)\colon \theta\in\Theta_N)$ are mutually orthogonal and with norm $N/2$.
\end{proof}

\begin{Lemma}\label{le:Hexp}
Under Assumption~\ref{ass:gauss} and $\mathcal{H}_0$ we have that $$ \|A(\vartheta_1,\ldots,\vartheta_q)A^\prime(\vartheta_1,\ldots,\vartheta_q)\|_{\mathrm{tr}}\sim \sum_{i=1}^{q} \Xi_{i},$$ where $\Xi_{1},\ldots, \Xi_{q}$ are i.i.d. random variables distributed as $\mathrm{HExp}(\lambda_1,\lambda_2,\ldots)$ and $(\lambda_k)$ are the eigenvalues of $\Gamma$ .
\end{Lemma}

\begin{proof}
Let $(v_k)$ be the eigenvectors of $\Gamma$. By Parseval's identity we have $\|R(\theta)\|^2=\sum_{k\geq 1}  \langle R(\theta),v_k\rangle^2$ and $\|C(\theta)\|^2=\sum_{k\geq 1}  \langle C(\theta),v_k\rangle^2$.
By Lemma~\ref{le:orth} it follows that $\langle R(\theta),v_k\rangle, \langle C(\theta),v_k\rangle$ are independent Gaussian with zero mean and variance $\lambda_k/2$. The result follows easily from $\|A(\vartheta_1,\ldots,\vartheta_q)A^\prime(\vartheta_1,\ldots,\vartheta_q)\|_{\mathrm{tr}}= \sum_{k=1}^q\big(\|R(\vartheta_k)\|^2 +\|C(\vartheta_k)\|^2\big)$.
\end{proof}

\medskip

\begin{proof}[Proof of Proposition~\ref{pr:anova_tr}]
We verify the identity in the functional case.
\begin{align*}
&\sum_{k=1}^{q} \left(\left\|R(\vartheta_{k})\right\|^{2} +\left\|C(\vartheta_{k})\right\|^{2}\right)\\
&\quad= \sum_{k=1}^{q}
\left(
\left\|\frac{n}{\sqrt{N}}\sum_{t=1}^{d} (\overline{Y}_{t}-\overline{Y}) \cos(\vartheta_{k} t)\right\|^{2}
+\left\|\frac{n}{\sqrt{N}}\sum_{t=1}^{d} (\overline{Y}_{t}-\overline{Y}) \sin(\vartheta_{k} t)\right\|^{2}
\right) \\
&\quad=
\frac{n}{d}\sum_{t=1}^{d}\sum_{s=1}^{d} \langle\overline{Y}_{t}-\overline{Y},\overline{Y}_{s}-\overline{Y}\rangle \sum_{k=1}^q\left(\cos(\vartheta_{k} t)\cos(\vartheta_{k}s)+\sin(\vartheta_{k} t)\sin(\vartheta_{k} s)\right)\\
&\quad=
\frac{n}{d}\sum_{t=1}^{d}\sum_{s=1}^{d} \langle\overline{Y}_{t}-\overline{Y},\overline{Y}_{s}-\overline{Y}\rangle \left(qI\{t=s\}-1/2I\{t\neq s\}\right)\\
&\quad=
\frac{n}{2}\sum_{t=1}^{d}\|\overline{Y}_{t}-\overline{Y}\|^2-\frac{1}{2}\sum_{t=1}^{d}\sum_{s=1}^{d} \langle\overline{Y}_{t}-\overline{Y},\overline{Y}_{s}-\overline{Y}\rangle\\
 &\quad=\frac{1}{2}\sum_{t=1}^{d} n\norm{\overline{Y}_{t}-\overline{Y}}^{2}.
\end{align*}
\end{proof}

\section{Proofs of the results of Sections \ref{s:dep}
and  \ref{s:asy} }\label{s:asy-p}

\begin{proof}[Proof of Proposition~\ref{pr:assfundep}]
According to Theorem~\ref{th:asympt} we
have that the vector of functions $(D_N^Z(\vartheta_{1}),\ldots, D_N^Z(\vartheta_{q}))$ will converge
weakly to a $q$--vector $\big(\mathcal{N}_1,\ldots,
\mathcal{N}_q\big)\in H^k$ where
$\mathcal{N}_k\stackrel{\text{ind}}{\sim}
\mathcal{CN}_H(0, \mathcal{F}_{\vartheta_{k}})$. Hence,
$\|R_N^Z(\vartheta_{k})\|^2+\|C_N^Z(\vartheta_{k})\|^2=\|D_N^Z(\vartheta_{k})\|^2\stackrel{d}{\to}\|\mathcal{N}_k\|^2$. For
any $\theta\in[-\pi,\pi]$ the operator $\mathcal{F}_\theta$ is
trace class, symmetric and non-negative definite. So it possesses the same
properties as a covariance operator. In particular we have a spectral
decomposition of the form $\mathcal{F}_\theta=\sum_{m\geq
1}\lambda_m(\theta)\varphi_m(\theta)\otimes\varphi_m(\theta)$, where
$\lambda_m(\theta)\geq 0$ are the eigenvalues (in descending order)
and $\varphi_m(\theta)$ are the corresponding (possibly complex
valued) eigenfunctions of $\mathcal{F}_\theta$.  By Parceval's
identity
\begin{equation} \label{e:proof42}
\|\mathcal{N}_k\|^2=\sum_{m\geq 1}\langle \mathcal{N}_k,\varphi_m(\vartheta_{k})\rangle\overline{\langle \mathcal{N}_k,\varphi_m(\vartheta_{k})\rangle}.
\end{equation}
The $(\varphi_m(\vartheta_{k})\colon m\geq 1)$ are, in fact, the principal components of $\mathcal{N}_k$. By their orthogonality, the normal scores $(\langle \mathcal{N}_k,\varphi_m(\vartheta_k)\rangle\colon m\geq 1)$ are independent $\mathcal{CN}\big(0,\lambda_m(\vartheta_k)\big)$. Hence,
$
(
\mathrm{Re}(\langle \mathcal{N}_k,\varphi_m(\vartheta_{k})\rangle)$,
$\mathrm{Im}(\langle \mathcal{N}_k,\varphi_m(\vartheta_{k})\rangle)
)\sim
(N_1,N_2)
$
where $N_1$ and $N_2$ are i.i.d. $N(0,\lambda_m(\vartheta_{k})/2)$. It implies that that $m$-th term of the sum in \eqref{e:proof42} is distributed as $\lambda_m(\vartheta_{k})\times \mathrm{Exp}(1)$.

\end{proof}

\begin{proof}[Proof of Proposition~\ref{pr:consistff}]

From the definition of $A(\vartheta_1,\ldots,\vartheta_q)$ and \eqref{eq:altern} we deduce
\begin{align*}
&\|A(\vartheta_1,\ldots,\vartheta_q)A^\prime(\vartheta_1,\ldots,\vartheta_q)\|_\mathrm{tr}=\sum_{j=1}^q\|D^Y_N(\vartheta_{j})\|^2\\
&\quad= \sum_{j=1}^q\left[\|D^Z_N(\vartheta_{j})\|^2+n\|D^w_d(\vartheta_{j})\|^2+2\sqrt{n}\,\mathrm{Re}\big(\langle D^w_d(\vartheta_{j}),D^Z_N(\vartheta_{j})\rangle\big)\right].
\end{align*}
When the noise satisfies the assumptions in Theorem~\ref{th:asympt},
then it follows that each term $\|D^Z_N(\vartheta_{j})\|^2=O_P(1)$.
\end{proof}

\medskip
\begin{proof}[Proof of Proposition~\ref{pr:allconsist}]
This goes along the lines of the proof of Proposition~\ref{pr:consistff}.
\end{proof}

\section{Estimation of covariance and spectral density}
\label{s:prac} To implement the
tests developed in Section~\ref{s:tests}, we have to estimate $\bSig$ and $\Gamma$. The tests of Section~\ref{s:dep} require
the estimation of a spectral density. In the following we will outline the estimation problem in the fully functional setting. The key point is to derive estimators which are consistent under $\mathcal{H}_0$ and $\mathcal{H}_A$. Failing to be consistent under $\mathcal{H}_A$, may still lead to consistent tests, but can have a strong negative impact on the power of the tests.

We recall that $C_h=\mathrm{Cov}(Y_{t+h},Y_t)$ is the lag-$h$ autocovariance operator of the time series $(X_t)$ and hence $C_0=\Gamma$. We stress that under the general model \eqref{e:model-G}, the process $(X_t)$ is covariance stationary.  Once we have estimators $\widehat w_t=\widehat w_{t+d}$ (we will require $\sum_{t=1}^d\widehat w_t=0$) it is natural to set
\begin{equation} \label{sigmaest}
\widehat{C}_h= \frac{1}{N} \sum_{t=1}^{N-h}
(Y_{t+h}-\widehat{w}_{t+h}-\overline Y)\otimes (Y_{t}-\widehat{w}_{t}-\overline Y).
\end{equation}
The following lemma translates the consistency rate for $\widehat{w}_t$ to a rate for $\widehat{C}_h$. We let $\|\cdot\|_\mathcal{S}$ be the Hilbert-Schmidt norm on the set of compact operators on $H$ and use $c_0,c_1,c_2$ for constants that do not depend on $N$ and $h$.
\begin{Lemma}\label{le:consistentCh}
Consider model \eqref{e:model-G}, Assumption~\ref{ass:d} and assume that $(Z_t)$ is $L^4$-$m$--approximable. Assume further that $N E\|\widehat w_t-w_t\|^2$ is bounded by some constant $b_0$ for all $t,N\geq 1$. Then
$$
E\|\widehat C_h- C_h\|_\mathcal{S}\leq c_0\sqrt{\frac{|h|\vee 1}{N}}.
$$
\end{Lemma}
The most simple estimator is $\widehat w_t=\overline{Y}_{t}-\overline{Y}$. Under $L^2$-$m$--approximability this estimator satisfies the consistency assumption of the lemma.
\begin{proof} Set $\nu_2(X)=\big(E\|X\|^2\big)^{1/2}$ and assume without loss of generality that $h\geq 0$. Let us first remark that $\overline{Y}=\mu+\overline{Z}$ and that basic properties for $L^2$-$m$--approximable sequences lead to
\begin{equation}\label{eq:mean}
N E\|\overline Z\|^2\leq 2\nu_2(Z_0)\sum_{k\geq 0}\nu_2(Z_0-Z_0^{(k)})=:c_1.\end{equation}
Then we decompose
$
(Y_{t+h}-\widehat{w}_{t+h}-\overline Y)\otimes (Y_{t}-\widehat{w}_{t}-\overline Y)-C_h
$
into
\begin{align*}
&[Z_{t+h}\otimes Z_t-C_h]+(w_{t+h}-\widehat{w}_{t+h})\otimes (w_{t}-\widehat{w}_{t}) +\overline Z\otimes \overline Z\\[1ex]
&\qquad-\big[\overline Z\otimes (w_t-\widehat{w}_t)+(w_{t+h}-\widehat{w}_{t+h})\otimes \overline{Z}\big]\\[1ex]
&\qquad-\big[\overline Z\otimes Z_t+Z_{t+h}\otimes \overline Z\big]\\
&\qquad+\big[(w_{t+h}-\widehat{w}_{t+h})\otimes Z_t+Z_{t+h}\otimes (w_t-\widehat{w}_t)\big]=:\sum_{j=1}^9A_t^{(j)}.
\end{align*}
For each $j\in \{1,\ldots,9\}$ we need a bound for $\kappa_j:=\frac{1}{N}E\|\sum_{t=1}^{N-h} A_t^{(j)}\|_\mathcal{S}$. Let us also introduce $\tilde\kappa_j:=\frac{1}{N}E\|\sum_{t=1}^{N} A_t^{(j)}\|_\mathcal{S}$ and note that $\kappa_j\leq \tilde\kappa_j+\frac{h}{N}\max E\|A_t^{(j)}\|_\mathcal{S}$.

By  Lemma~4 in \citetext{hormann:kidzinski:hallin:2015} we have that $\kappa_1\leq c_2 \sqrt{\frac{h\vee 1}{N}}$. Then since
$$\sum_{t=1}^N A_t^{(2)}=n\sum_{k=1}^d(w_{k+h}-\widehat{w}_{k+h})\otimes (w_{k}-\widehat{w}_{k}),$$ we have that $\tilde\kappa_2\leq \frac{b_0}{N}$ and thus $\kappa_2\leq \frac{b_0}{N}(1+\frac{h}{N})$.
By \eqref{eq:mean} we have $\kappa_3=E\|\overline Z\|^2\leq \frac{c_1}{N}$.
Next, by the assumptions on $w_t$ and $\widehat w_t$ we have $\tilde\kappa_4=\
\tilde\kappa_5=0$ and hence $\kappa_4,\kappa_5\leq \sqrt{\frac{c_1}{N}\frac{b_0}{N}}\big(1+\frac{h}{N}\big).$ Again by \eqref{eq:mean} it follows that $\kappa_6,\kappa_7\leq\frac{c_1}{N}$. Finally, it can be shown along the same lines that $\kappa_8,\kappa_9\leq \sqrt{\frac{dc_1}{N}\frac{b_0}{N}}$. The proof follows by collecting all terms.
\end{proof}

A simple implication of this lemma is that whenever we have $L^2$ consistent estimators for $w_t$, then $\widehat\Gamma=\widehat C_0$ is a consistent estimator of $\Gamma$.

Turning to the estimation of the spectral density, we propose a lag window estimator of the form
\begin{equation} \label{e:win2}
\mathcal{\widehat F}_\theta
= \sum_{|h|\leq b_{N}}\gamma\left (\frac{h}{b_{N}}\right )
\widehat{C}_h e^{-\mathrm{i} h\theta},\quad 0<b_{N}<N,
\end{equation}
where the function $\gamma$ is continuous in zero and such that $\gamma(0)=1$, $\abs{\gamma(x)}\leq 1,\;
\forall x$ and $\gamma(x)=0$ for $\abs{x}>1$. The bandwidth satisfies
$b_{N}\to \infty$ and ${b_{N}}/{N}\to 0$, with  more specific
rate stated in the next proposition.

\begin{Proposition}\label{lem:gammas}
Consider the setup of Lemma~\ref{le:consistentCh} and suppose that
  $b_N \to\infty$, such that $b_N=o\big(N^{1/3}\big)$. Then \
  $\sup_{\theta\in[-\pi,\pi]}\|\mathcal{F}_\theta-\mathcal{\widehat{
    F}}_\theta\|_\mathcal{S}\to 0$ in probability.
\end{Proposition}
\begin{proof}
We will adapt the proof of Proposition~4 of \citetext{hormann:kidzinski:hallin:2015} to our setting. By definition of the estimator and by using the triangular inequality, we have
 \begin{align*}
\norm{\mathcal{F}_{\theta}-\widehat{\mathcal{F}}_{\theta}}_{\mathcal{S}} &= \left\|\sum_{h\in \mathbb{Z}} C_{h} e^{-ih \theta} - \sum_{\abs{h} \leq b_N}\gamma\left (\frac{h}{b_{N}}\right ) \widehat{C}_{h} e^{-i\theta h}\right\|_{\mathcal{S}}\\
&\leq \sum_{\abs{h} \leq b_N}\left(\|C_h-\widehat{C}_{h}\|_\mathcal{S}+|1-\gamma(h/b_N)|\|C_h\|_\mathcal{S}\right)+\sum_{|h|>b_N} \|C_h\|_\mathcal{S}.
 \end{align*}
Taking expectations, we obtain by Lemma~\ref{le:consistentCh}
$$
E\norm{\mathcal{F}_{\theta}-\widehat{\mathcal{F}}_{\theta}}_{\mathcal{S}} \leq 2c_0 \sqrt{\frac{b_N^3}{N}}+\sum_{\abs{h} \leq b_N} |1-\gamma(h/b_N)|\|C_h\|_\mathcal{S}+\sum_{|h|>b_N} \|C_h\|_\mathcal{S}.
$$
By our assumptions, all three terms on the right hand side tend to zero. For the first, this is immediate from our assumption on $b_N$. For the third we use \eqref{e:abssymcov} and for the second dominated convergence.
 \end{proof}
It is an immediate consequence of Proposition~\ref{lem:gammas} and
Corollary 1.6 on p.\ 99  of \citetext{gohberg:1990}
 that
\[
\sup_{\theta\in [-\pi,\pi]}\sup_{m\geq
1}|\lambda_m(\theta)-\hat\lambda_m(\theta)|=o_P(1),
\]
justifying the approximation of the model eigenvalues by the
estimated eigenvalues.

\section{Local consistency} \label{s:local}

In this section we focus for brevity only on the
functional test statistics $T^{\tt FTR_1}$ and $T^{\tt FAV}$ (or equivalently $T^{\tt FTR_2}$).  The multivariate tests require additional tuning  (choice of the basis). For a fixed basis the discussion below can be done analogously.

We consider consistency of the test, when the alternative shrinks to the null with growing sample size.
An important finding of this section is that if the period $d$ is relatively large, and if we
expect a smooth change for the periodic trend, the test based on
 $T^{\tt FTR_1}$ is preferable to
the ANOVA approach. In contrast, if the periodic signal is more erratic, the ANOVA approach has better local power features. This confirms simulation results in Section~\ref{ss:localpower}.

More specifically we consider here local alternatives $(w_t (u)\colon
1\leq t\leq d)$, with \emph{mean square sum of the signal}
$\text{MSS}_{\text{sig}}=\frac{1}{d}\sum_{k=1}^d\|w_k\|^2=\varrho^2\to 0$.
We analyze three scenarios for
$\mathcal{H}_A$. In each we define $w_t=\omega_t-\bar{\omega}$ with
$\bar{\omega}=\frac{1}{d}\sum_{k=1}^d\omega_k$. Then we consider for some
$\varrho^2>0$ the following:
\begin{enumerate}[(A)]
\item $\omega_t$ are orthogonal functions with
$\|\omega_t\|^2=\frac{d}{d-1}\varrho^2$\,;
\item $\omega_t=\omega_0\big[\cos(2\pi t/d)
+\sin(2\pi t/d)\big]$ with $\|\omega_0\|^2=\varrho^2$\,;
\item $\omega_t=\frac{1}{\sqrt{d}}\big(\sum_{k=1}^t v_k-\frac{t}{d}\sum_{k=1}^d v_k\big)$ for orthogonal functions $v_k\in H$ with scaling $\norm{v_k}^2=\frac{12d^2}{d^2-1}\varrho^2$.
\end{enumerate}
Straightforward computations show that indeed for (A)--(C) we have
$\text{MSS}_{\text{sig}}=\varrho^2.$

In setting (A) the periodic pattern is extremely irregular and
intuitively it should be very unfavorable for the test based on
 $T^{\tt FTR_1}$ since there is no
sinusoidal trend involved. In contrast, scenario (B) is tailor-made
for this test. Finally, scenario~(C) is supposed to provide a realistic alternative. It is based on the assumption that
the periodic trend $(w_t)$ changes smoothly. Orthogonality of the
$v_k$ is only imposed to make $\text{MSS}_{\text{sig}}$
computable.

Next we define the power functions
\begin{equation*}\label{eq:pvalue}
\kappa_1(\alpha)=P\big(T^{\tt FTR_1}>q_{1-\alpha}^{(1)}\big)\quad\text{and}\quad
\kappa_2(\alpha)=P\big( T^{\tt FAV}>q_{1-\alpha}^{(2)}\big),
\end{equation*}
where $q_{1-\alpha}^{(1)}$ and $q_{1-\alpha}^{(2)}$ are the
$(1-\alpha)$--quantiles of the respective null-distributions.

\begin{Proposition} \label{pr:localpower}
Let $\alpha\in (0,1)$ and assume that $\varrho^2=\varrho_n^2\to 0$ with $n\to\infty$. We impose Assumption~\ref{ass:gauss}.
Then the following claims hold.

\begin{enumerate}
\item When $n\varrho_n^2\to\infty$,
both power functions $\kappa_1(\alpha)$ and
$\kappa_2(\alpha)$ tend to one under all three alternatives.
\item When $n\varrho_n^2\to 0$ both power functions $\kappa_1(\alpha)$
and $\kappa_2(\alpha)$ tend to $\alpha$ under all three
alternatives.
\item Suppose $d=d_n\to\infty$ and $n\varrho_n^2\to 0$.
Then $\kappa_2(\alpha)\to \begin{cases}1&
\text{if $\sqrt{d_n}n\varrho_n^2\to \infty$;}\\
\alpha &\text{if $\sqrt{d_n}n\varrho_n^2\to 0$}.\end{cases}$
\item Suppose $d=d_n\to\infty$ and $n\varrho_n^2\to 0$. Under alternative (A) $\kappa_1(\alpha)\to \alpha$ while under alternatives (B) and  (C) we have that $\kappa_1(\alpha)\to \begin{cases}1& \text{if $d_nn\varrho_n^2\to \infty$;}\\
\alpha & \text{if $d_nn\varrho_n^2\to 0$.}\end{cases}$
\end{enumerate}

\end{Proposition}

This proposition shows that whether one of the tests is asymptotically
consistent or not, is determined by $n\varrho_n^2\to\infty$ and $n\varrho_n^2\to 0$,
respectively. The interesting case is when $n\varrho_n^2\to 0$ and at the same time
$d=d_n\to\infty$. Here under setup (A) the statistic $T^{\tt FAV}$ is
preferable since it can still provide a consistent test. On the other
hand, statement~\emph{4} of the proposition shows superiority
of  $T^{\tt FTR_1}$ under a local
alternative related to settings (B) and (C). In these cases $T^{\tt FTR_1}$ allows additional shrinking the alternative by a factor $1/\sqrt{d_n}$ compared $T^{\tt FAV}$ in order to remain consistent.

\begin{proof}[Proof of Proposition~\ref{pr:localpower}]
Denote by $\| M^{Y}_{N}(\vartheta_1) \|_{\mathrm{tr}}, \; \|
M^{Z}_{N}(\vartheta_1) \|_{\mathrm{tr}}$ and
$\| M^{w}_{d}(\vartheta_1)
\|_{\mathrm{tr}}$  the trace statistics based on the time series
$Y_{t}$, $Z_{t}$  $(t=1,\ldots,N)$ and $w_{t}$ $(t=1, \ldots, d)$. From \eqref{eq:altern} one
easily deduce that \[\| M^{Y}_{N}(\vartheta_1) \|_{\mathrm{tr}}=\|
M^{Z}_{N}(\vartheta_1) \|_{\mathrm{tr}}+n \|
M^{w}_{d}(\vartheta_1) \|_{\mathrm{tr}}+O_P\Big(\big(n
\| M^{w}_{d}(\vartheta_1)
\|_{\mathrm{tr}}\big)^{1/2}\Big).\] The trace statistics of the
periodic signal for alternatives (A), (B) and (C) are, respectively,
$\frac{d}{d-1}n\varrho_n^2$, $dn\varrho^2_n$ and $\frac{3dn\varrho^2_n
}{(d^{2}-1) \sin^{2}(\pi/d)}$. Consider, for example, alternative
(A). Then
\begin{align*}
P(\| M^{Y}_{N}(\vartheta_1) \|_{\mathrm{tr}}>q_{1-\alpha}^{(1)})&=P\left(\| M^{Z}_{N}(\vartheta_1) \|_{\mathrm{tr}}>q_{1-\alpha}^{(1)}-\frac{d}{d-1}n\varrho^2_n+O_P\bigg(\sqrt{\frac{d}{d-1}n\varrho^2_n}\bigg)\right).
\end{align*}
To verify claims \emph{1,2} and \emph{4} related to the trace based
statistic, it suffices to observe that this probability tends to
$\alpha$ if $n\varrho^2_n\to 0$ and to 1 if $n\varrho^2_n\to
\infty$. The statements under alternatives (B) and (C$^\prime$) are
proven analogously.

Concerning statistic $T^{\tt FAV}$  we first note that by the assumption
$\mathrm{MSS}_\mathrm{sig}=\varrho_n^2$ it readily follows that
\[
\frac{n}{d}\sum_{k=1}^d\langle w_k,\overline{Z}_k-\overline{Z}\rangle=\frac{n}{d}\sum_{k=1}^d\langle w_k,\overline{Z}_k\rangle =O_P\big(\sqrt{n\varrho_n^2/d}\big).
\]
With this one can easily prove that
\[
P( T^{\tt FAV} >q_{1-\alpha}^2)=P\Big( T^{\tt FAV}(Z) >q_{1-\alpha}^2-n\varrho_n^2+O_P\big(\sqrt{ n\varrho_n^2/d}\big)\Big).
\]
Here $T^{\tt FAV}(Z)$ is the ANOVA statistic computed from the noise.
Claims \emph{1} and \emph{2} are immediate.

We show claim \emph{3}. Let us impose for simplicity that Assumption~\ref{ass:d} holds. Then, since by Gaussianity the distribution of $T^{\tt FAV}(Z)$ is independent of $n$, we get by Corollary~\ref{c:anova} that

\begin{align*}
&P( T^{\tt FAV}(Y)>q_{1-\alpha}^2)=P\Big(\frac{2}{d}\sum_{k=1}^q\Xi_k> q_{1-\alpha}^\text{FAV} - n\varrho_n^2+O_P(\sqrt{n\varrho_n^2/d})\Big)\\
&\,=P\left(\frac{1}{\sqrt{q}}\sum_{k=1}^q(\Xi_k-E\Xi_1)> q_{1-\alpha}\left[\frac{1}{\sqrt{q}}\sum_{k=1}^q(\Xi_k-E\Xi_1)\right] -\frac{d n\varrho_n^2}{2\sqrt{q}}+O_P(\sqrt{n\varrho_n^2d/q})\right).
\end{align*}
By the central limit theorem $\frac{1}{\sqrt{q}}\sum_{k=1}^q(\Xi_k-E\Xi_1)\stackrel{d}{\to} N(0,\mathrm{Var}(\Xi_1))$ and hence the quantiles $q_{1-\alpha}\left[\frac{1}{\sqrt{q}}\sum_{k=1}^q(\Xi_k-E\Xi_1)\right]$ will converge to the corresponding normal quantile for $d\to\infty$. From this it is easy to conclude.
\end{proof}

\end{document}